%% file: main.tex
\documentclass[11pt,letterpaper]{article}
\usepackage[margin=1in]{geometry}
\usepackage{setspace}
\usepackage{comment}
\usepackage{amsmath,amssymb,amsthm}
\usepackage{hyperref}
\usepackage{cleveref}
\usepackage{mathrsfs}
\usepackage{mdframed}

\usepackage[linesnumbered,ruled,vlined]{algorithm2e}

\usepackage{graphicx,subcaption}
\usepackage{enumitem}
\usepackage{framed}
\usepackage[colorinlistoftodos, textsize=scriptsize]{todonotes}
\newcommand{\tgraphic}[1]{\includegraphics[width=0.75\linewidth, trim=0cm 1.5cm 0cm 1.5cm, clip]{#1}}

\newcommand{\BEDec}[1]{\operatorname{BE-Decomp}(#1)}
\newcommand{\VSet}{\operatorname{V}}
\newcommand{\ESet}{\operatorname{E}}

\newcommand{\nsp}{\textsf{NextSP}}
\newcommand{\nspl}{\textsf{NextSP-Layered}}
\newcommand{\nsps}{\textsf{NextSP-Straight}}

\newcommand{\ignore}[1]{}
\newcommand{\utov}{\!\!\to\!\!}
\newcommand{\tightto}[2]{#1\utov #2}
\newcommand{\stot}{s\utov t}
\newcommand{\anyP}{\mathscr{P}}

\newcommand{\downV}{V^{\operatorname{down}}}
\newcommand{\upV}{V^{\operatorname{up}}}
\newcommand{\gapV}{V^{\operatorname{gap}}}
\newcommand{\interV}{V^{\operatorname{inter}}}
\newcommand{\calI}{\mathcal{I}}

\usepackage{xparse} 
\usepackage{thm-restate}
\newtheorem{theorem}{Theorem}[section]
\newtheorem{theorem*}{Theorem}
\newtheorem{lemma}[theorem]{Lemma}
\newtheorem*{lemma*}{Lemma}
\newtheorem*{keylemma*}{Key Lemma}

\theoremstyle{definition}
\newtheorem{definition}[theorem]{Definition}
\newenvironment{proofoflem}[1] 
  {\vspace{0.2cm}
  \noindent 
  {\em Proof of Lemma #1.}\mbox{}} 
  {\hfill$\square$} 
\newenvironment{proofofthm}[1] 
  {\vspace{0.2cm}
  \noindent 
  {\em Proof of Theorem #1.}\mbox{}} 
  {\hfill$\square$} 
  
\newenvironment{proofwithassumption}[2]
  {\vspace{0.2cm}
  \noindent 
  {\em Proof of #1 (given #2).}\mbox{}} 
  {\hfill$\square$\vspace{0.2cm}} 
\title{A Polynomial-Time Algorithm for the Next-to-Shortest Path Problem on Positively Weighted Directed Graphs}
\author{Kuowen Chen\footnote{This work was initiated while the author was at the University of Michigan. \texttt{ckw22@mails.tsinghua.edu.cn}}\\
{\small IIIS, Tsinghua University}
\and
Nicole Wein\footnote{\texttt{nswein@umich.edu}}\\
{\small University of Michigan}
\and 
Yiran Zhang\footnote{\texttt{zhangyir22@mails.tsinghua.edu.cn}}\\
{\small IIIS, Tsinghua University}
}
\date{}

\begin{document}

\maketitle
\begin{abstract}

Given a graph and a pair of terminals $s$, $t$, the \emph{next-to-shortest path} problem asks for an $s\utov t$ (simple) path that is shortest among all \emph{not} shortest $s\utov t$ paths (if one exists). This problem was introduced in 1996, and soon after was shown to be NP-complete for directed graphs with non-negative edge weights, leaving open the case of positive edge weights. Subsequent work investigated this open question, and developed polynomial-time algorithms for the cases of undirected graphs and planar directed graphs. In this work, we resolve this nearly 30-year-old open problem by providing an algorithm for the next-to-shortest path problem on directed graphs with positive edge weights.

\end{abstract}
\pagenumbering{gobble}
\clearpage
\pagebreak
\pagenumbering{arabic}
\input{intro}

\input{prel}
\input{reduction}

\input{algorithm}

\input{key-lemma}


\section*{Acknowledgments}
We would like to thank Luke Schaeffer and Fabrizio Grandoni for making us aware of the problem. We would like to thank the Bertinoro 2019 Fine-Grained Approximation Algorithms \& Complexity Workshop for initial discussions about the problem, including discussions with Andreas Bj\"orklund, Thore Husfeldt, Vijaya Ramachandran, and Virginia Vassilevska Williams. We are very grateful to Shyan Akmal, Jacob Holm, Eva Rotenberg, and Virginia Vassilevska Williams for extended discussions about the problem in 2020. 

More recently, we would like to thank the participants of the open problem sessions of Nicole Wein's \emph{Graph Algorithms} course for additional discussions about the problem. In particular, we would like to thank participants Zixi Cai, Shengquan Du, Bing Han, Teo Miklethun, Benyu Wang, and Qingyue Wu. We also thank Erik Demaine for use of his ``supercollaboration'' method and accompanying \emph{Coauthor} software, for the open problem sessions. We would also like to thank Yang Hu and Junkai Zhang for additional discussions. Finally, we would like to thank Seth Pettie for organizing the semester-long visit of Tsinghua students to the University of Michigan.

\bibliographystyle{alpha} 
\bibliography{references}
\end{document}

%% file: intro.tex
\section{Introduction}

Given a directed graph with positive edge weights, and a pair of terminals $s$, $t$, the \emph{next-to-shortest path} problem (also known as the \emph{strictly second-shortest path} problem) asks for an $s\utov t$ (simple) path that is shortest among all \emph{not} shortest $s\utov t$ paths (if one exists). We present the first polynomial time algorithm for this problem. 

The next-to-shortest path problem was first introduced by Lalgudi, Papaefthyrniou, and Potkonjak in 1996 for the application of processing data streams in hardware~\cite{lalgudi1996optimizing,lalgudi2000optimizing}. Soon after, Lalgudi and Papaefthymiou proved that the problem is NP-complete on directed graphs with non-negative edge weights~\cite{lalgudi1997computing}, but that there is a polynomial-time algorithm for the relaxation where the path need not be simple. In their NP-hardness reduction, the weight-0 edges play a crucial role, and the case of positive edge weights remained open for almost 30 years until the present work. In the meantime there was a long line of work on the next-to-shortest path problem in \emph{undirected} graphs. 

In 2004, Krasikov and Noble gave the first polynomial-time algorithm for the undirected next-to-shortest path problem with positive edge weights~\cite{krasikov2004finding}. (They also conjectured that it is NP-complete for directed graphs; we disprove this conjecture if $P \not = NP$.) For an $n$-vertex, $m$-edge graph, their algorithm ran in time $O(n^3m)$ time, and was subsequently improved by Li, Sun, and Chen to $O(n^3)$~\cite{li2006improved}, then by Kao, Chang, Wang, and Juan to $O(n^2)$~\cite{kao2011quadratic}, and finally by Wu to $O(n \log n + m)$ with the additional guarantee that the path can be found in linear time if the distances from $s$ and $t$ to all other vertices have already been computed~\cite{wu2013simpler}. This was extended to non-negative edge weights by Wu, Guo, and Wang~\cite{wu2012linear} (see also independent work by Zhang and Nagamochi with larger running time~\cite{zhang2012next}). 
Before the general solution was known, the problem was also studied on special classes of graphs~\cite{mondal2006sequential,barman2009efficient}. See also~\cite{demaine2014canadians} for a data structure that implicitly represents near-shortest $s$-$t$ paths. Additionally, the problem of finding an \emph{induced not-shortest} path  is also solvable in polynomial time on unweighted graphs~\cite{berger2021finding}. 

For \emph{directed} graphs with positive edge weights, Wu and Wang gave a polynomial-time algorithm of the next-to-shortest path problem when the graph is planar~\cite{DBLP:journals/networks/WuW15}. They also proved some properties of a solution for general positively weighted directed graphs. However no further progress was made, and in fact, it was not even known how to find \emph{any} not-shortest $s\utov t$ path in polynomial time.

At least 11 papers have explicitly stated the open problem solved in this paper: 
\begin{center} \emph{Is the next-to-shortest path problem (or even the not-shortest path problem) on positively weighted directed graphs solvable in polynomial time?}
\end{center} 
In these papers~\cite{krasikov2004finding,kao2011quadratic, wu2012linear,wu2013simpler, DBLP:journals/networks/WuW15,dolgui2018simple,akmal2022local,akmal2023faster,hatzel2023simpler,fomin2023detours,assadi2025covering}, the problem has been called ``important", ``interesting", ``surprising", and ``arguably one of the most tantalizing open questions within the area of graph algorithms."

\subsection{Our Contribution}
We provide a polynomial-time algorithm for the next-to-shortest path problem on directed graphs with positive edge weights:
\begin{restatable}{theorem}{mainthm}\label{thm:main}
Given a directed graph $G=(V,E,w)$, where $w: E\rightarrow \mathbb{R}_{>0}$ assigns the edge weights, and two vertices $s,t\in V$, there exists an $O(|V|^4|E|^3\log |V|)$-time algorithm that determines if there exists a next-to-shortest path from $s$ to $t$, and if so, outputs such a path. 
\end{restatable}

Our focus was to obtain a polynomial-time algorithm rather than to optimize the running time. We introduce some inefficiencies into the algorithm for the sake of clarity. For example, by modifying the algorithm to not subdivide edges, we believe it could be possible to get time $O(|V|^4|E|^2\log |V|)$, but getting a much faster algorithm is an interesting open question.

Our result opens the door to studying generalizations of the problem that have received attention in undirected graphs and special classes of directed graphs, but have been out-of-reach for general directed graphs. The most prevalent example is the \emph{$k$-longest detour} problem from parameterized algorithms. The input is a graph, specified vertices $s,t$, and a positive integer $k$, and the goal is to determine if $G$ has an $s\utov t$ path of length at least $d(s,t)+k$. The not-shortest path problem is the special case where $k=1$. The $k$-longest detour problem is known to be fixed-parameter tractable for both undirected graphs~\cite{bezakova2019finding,akmal2023faster} and directed planar graphs~\cite{fomin2023detours,hatzel2023simpler}.  Our techniques may lead to a solution to this problem in general directed graphs.

\subsection{Related Work} 

\paragraph{Exact Detours.} The \emph{$k$-exact detour} problem is the same as the above $k$-longest detour problem, except the path must be of length \emph{exactly} $d(s,t)+k$. This problem is known to be fixed-parameter tractable for general graphs, both undirected and directed~\cite{bezakova2019finding,akmal2023faster}. 

\paragraph{Paths of exact and forbidden lengths.}
The problem of finding an $s\utov t$ path of an exact given length has been studied~\cite{nykanen2002exact,salmela2016gap}, as well as the more general problem of finding paths whose length falls into a specified range~\cite{dolgui2018simple}.

\paragraph{Disjoint Paths.} The next-to-shortest path problem is closely related to disjoint path problems. Indeed, ensuring that a next-to-shortest path is simple requires the segments of the path to be disjoint from each other.  For example, the NP-completeness proof for the next-to-shortest path problem in directed graphs with non-negative edge weights~\cite{lalgudi1997computing} uses a simple reduction from the 2-disjoint paths problem~\cite{FORTUNE1980111} (given a directed graph and terminals $s_1,t_1,s_2,t_2$, find an $s_1\utov t_1$ path and an $s_2\utov t_2$ path that are disjoint). As another example, our algorithm employs a subroutine for solving the 2-disjoint paths problem in directed acyclic graphs~\cite{tholey2012linear}. 

The $k$-disjoint paths problem has been studied extensively on undirected graphs~\cite{lynch1975equivalence,middendorf1993complexity,robertson1995graph,kawarabayashi2012disjoint}, directed graphs~\cite{FORTUNE1980111}, directed acyclic graphs~\cite{fleischer2007efficient,FORTUNE1980111,slivkins2010parameterized,chitnis2023tight}, directed planar graphs~\cite{schrijver1994finding,6686155}, and undirected planar graphs~\cite{lokshtanov2020exponential,10353213,lokshtanov2020exponential,cho2023parameterized}.

There is also a wealth of work on disjoint paths problems with various objective functions such as \emph{min-sum} (minimizing the sum of the lengths of the paths), \emph{min-max}, and \emph{min-min}. Additionally, there is a literature on the \emph{disjoint shortest paths} problem where each of the paths must individually be shortest.

\subsection{Organization}
In~\Cref{sec:tech}, we provide a technical overview. \Cref{sec:prelim} is the preliminaries. In~\Cref{sec:reduction}, we present a reduction from a positively weighted graph to a special class of graphs, which we call \emph{$(s,t)$-layered} graphs.  
In~\Cref{sec:algorithm}, we state the Key Lemma (\Cref{lem:key}) without proof and, based on it, present a polynomial-time algorithm for the next-to-shortest path problem.  
Finally, in~\Cref{sec:proof-of-key-lemma}, we prove the Key Lemma.

\section{Technical Overview}\label{sec:tech}

First (\Cref{sec:reduction}) we reduce to the case of \emph{$(s,t)$-layered} graphs, which we define in~\Cref{def:layered-graph}. Informally, (1) every vertex is on some shortest $s\utov t$ path, and (2) if we partition the vertices into layers based on their distance from $s$, every edge either goes exactly one layer forward (forward-edge), or an arbitrary number of layers backward (back-edge). This reduction requires careful treatment of technical details, but it is not the main conceptual idea of the proof.

Instead, the main idea is captured by the Key Lemma (\Cref{lem:key}), which can be informally stated as follows: 
\begin{keylemma*}[Informal]
    Given an $(s,t)$-layered graph, fix a next-to-shortest path $P^*$, and let $P_0^*$ be the subpath from first to last back-edge, denoting the endpoints of $P_0^*$ as $A,B$. Then, there exists a pair of forward edges $(X',X)$, $(Y',Y)$ on the same layer as each other, so that \emph{all} pairs of vertex-disjoint shortest paths $P_1,P_2$ of the form $s \utov X' \utov X \utov A$ and $B\utov Y'\utov Y \utov t$ (respectively) are internally vertex-disjoint from $P_0^*$.
\end{keylemma*}

Intuitively, the Key Lemma is powerful because it yields a polynomial-time algorithm (\Cref{sec:algorithm}) that tries all possible choices of $A$, $B$, $(X',X)$, $(Y',Y)$, and works roughly as follows. For each choice, pick an \emph{arbitrary} pair $P_1$, $P_2$ of disjoint shortest paths of the form $s \utov X'\utov X\utov A$ and $B\utov Y'\utov Y\utov t$, respectively. It is not difficult to show that this is possible by utilizing a known linear-time algorithm for finding 2 disjoint paths in a directed acyclic graph~\cite{tholey2012linear} (building on prior work~\cite{shiloach1978finding,lucchesi1992irrelevance}). After computing $P_1,P_2$, remove their edges from the graph and take $P_0$ to be a shortest $A \utov B$ path in the remaining graph. Then return the concatenation $P_1\circ P_0\circ P_2$. The Key Lemma guarantees that it suffices to first fix $P_1,P_2$ before finding $P_0$ since \emph{all} such $P_1,P_2$ are guaranteed to be disjoint from $P_0$, so the resulting path is indeed a simple path.

Conceptually, the main part of the overall proof is to prove the Key Lemma. Towards doing so, we suppose for contradiction that $P_1 \cup P_2$ intersects $P_0^*$. Ultimately, the contradiction comes from constructing an $s\utov t$ not-shortest path $P$ that is \emph{shorter} than $P^*$, which contradicts the fact that $P^*$ is a \emph{next-to-shortest} path. To construct such a path $P$, we will stitch together segments of $P_1$, $P_2$, and $P^*$ (and other paths that we show exist). In doing so, the assumption that $P_1 \cup P_2$ intersects $P_0^*$ is important since these intersection points enable the concatenation of segments of different paths. 

In~\Cref{subsec:pdfps}, we consider how $P_1 \cup P_2$ intersects with the other parts of $P^*$ besides $P_0^*$, since such intersection points will also be useful to construct $P$: Let $P_1^*$ be the segment of $P^*$ from $s$ to $A$ and let $P_2^*$ be the segment of $P^*$ from $B$ to $t$. Intuitively, using the fact that $P_1$ and $P_2$ are \emph{shortest} paths, we conclude that neither $P_1$ nor $P_2$ can intersect $P_1^*$ and then subsequently intersect $P_0^*$ because this would constitute a ``shortcut'' along $P^*$ which would result in a not-shortest path that would be shorter than $P^*$. For a similar reason, neither $P_1$ nor $P_2$ can intersect $P_0^*$ and then subsequently intersect $P_2^*$. These properties provide structure that aids the construction of $P$.

Additionally, for technical reasons, it is important to establish \emph{multiple} intersection points between $P_1 \cup P_2$ and $P_0^*$  to provide flexibility when constructing $P$. To do so, we need to choose the pair of edges $(X',X)$, $(Y',Y)$ carefully. A key definition (in \Cref{subsec:construction}) is that $(X',X)$ and $(Y',Y)$ are defined to be edges on $P_1^*$, $P_2^*$ respectively, so that: 
\begin{enumerate}
    \item there exist $s \utov X$ and $B \utov Y$ disjoint shortest paths whose union intersects $P_0^*$, and
    \item there do \emph{not} exist $s \utov X'$ and $B \utov Y'$ disjoint shortest paths whose union intersects $P_0^*$.
\end{enumerate}

Item 1 establishes an intersection point $v$ between some $P_1 \cup P_2$ and $P_0^*$, where $v$ appears before the layer of $X,Y$. Item 2 establishes another intersection point: Consider shortest disjoint paths $s \utov X' \utov X \utov A$ and $B \utov Y' \utov Y \utov t$. If no such paths intersect $P_0^*$ then we are already done, and if they do, then the intersection point must appear \emph{after} the layer of $X,Y$ by the condition of item 2. These two intersection points, before and after the layer of $X,Y$, are critical in constructing the path $P$. 

After establishing the above properties and intersection points, in~\Cref{subsec:final} we construct $P$ by considering two cases based on whether the first point along $P_0^*$ that is shared with $P_1 \cup P_2$ appears in a layer before or after $X,Y$. For each of these two cases, we stitch together path segments in two different ways. To do so, we use both the above structural properties of how $P_1,P_2$ can intersect with $P^*$, and the above pair of intersection points between $P_1 \cup P_2$ and $P_0^*$. The final path $P$ is the concatenation of 7 or 9 path segments (in each of the two cases respectively), and is depicted in~\Cref{fig:final-1} and~\Cref{fig:final-2}.

%% file: prel.tex
\section{Preliminaries}\label{sec:prelim}
\paragraph{Graphs}
In this paper, we only consider directed graphs with positive edge weights. We define a graph as a triple \( G = (V, E, w) \), where \( V \) is the set of vertices, \( E \subseteq V \times V \) is the set of edges, and \( w : E \rightarrow \mathbb{R}_{>0} \) assigns a positive weight to each edge. Without loss of generality, throughout this paper, all graphs are assumed to be simple directed graphs, i.e., they contain no self-loops or multiple edges, unless stated otherwise.

When the graph is \emph{unweighted}, we treat it as a special case where all edge weights are equal to 1; that is, \( w(e) = 1 \) for all \( e \in E \). In this case, we abbreviate the graph as \( G = (V, E) \), omitting the weight function.

\paragraph{Walks and Paths}

Let \( G = (V, E, w) \) be a directed graph, where \( w : E \to \mathbb{R}_{> 0} \) assigns positive weights to edges.  
A \emph{walk} in \( G \) is a finite sequence of vertices
\[
W = (v_0, v_1, \dots, v_k)
\]
such that \( (v_{i-1}, v_i) \in E \) for all \( i = 1, \dots, k \). A walk may contain repeated vertices and edges. 

We define the following notation:
\begin{itemize}
  \item \textbf{Length (or Cost)}: $w(W):=\sum_{1\le i\le k} w(v_{i-1},v_i)$ denotes the sum of edge weights in $W$; 
  \item \( \VSet(W) := \{ v_0, v_1, \dots, v_k \} \) denotes the set of vertices appearing in the walk \( W \);
  \item \( \ESet(W) := \big\lbrace (v_0, v_1), (v_1, v_2), \dots, (v_{k-1}, v_k) \big\rbrace \) denotes the set of edges in \( W \);
\end{itemize}
For two walks \( W_1 = (u_0, u_1, \dots, u_k) \) and \( W_2 = (u'_0, u'_1, \dots, u'_\ell) \),  
if \( u_k = u'_0 \), we define their \textbf{concatenation} as
\[
W_1 \circ W_2 := (u_0, u_1, \dots, u_k, u'_1, \dots, u'_\ell),
\]
which is a walk formed by appending \( W_2 \) to \( W_1 \) at the shared vertex \( u_k = u'_0 \). In addition, we say $W_1$ and $W_2$ are \textbf{fragments} of $W_1\circ W_2$ in the context of concatenation.

A \textbf{path} is a walk in which all vertices are distinct, i.e., \( v_i \neq v_j \) for all \( i \neq j \), and hence all edges are distinct as well.

Given a path \( P = (v_0, v_1, \dots, v_k) \), a \textbf{subpath} is any contiguous subsequence of vertices.  
For two vertices \( x = v_i \) and \( y = v_j \) with \( 0 \le i < j \le k \), we write \( P_{x \to y} := (v_i, v_{i+1}, \dots, v_j) \) to denote the subpath of \( P \) from \( x \) to \( y \).
\paragraph{Distance} Consider a directed graph $G=(V,E,w)$. For any two vertices \( u, v \in V \), we call a $\stot$ path $P$ shortest iff the cost of $P$ is minimum among all paths from $s$ to $t$. We define the distance \( d_G(u, v) \) as the cost of the shortest path from $u$ to $v$, i.e.,
\[
d_G(u, v) := \min\{ w(P) : P \text{ is a path from } u \text{ to } v \}.
\]

Specially, if there is no path from $u$ to $v$, we denote $d_{G}(u,v)=\infty$. Since all edge weights are positive, the shortest-path distance function \( d_G(\cdot, \cdot) \) satisfies the triangle inequality:
\[
d_G(u,v) + d_G(v,w) \geq d_G(u,w), \quad \text{for all } u,v,w \in V.
\]
We will use this property freely throughout the paper.
\paragraph{Not-Shortest Path and Next-to-Shortest Path}

Now we show the definition of not-shortest path and next-to-shortest path.

\begin{definition}[$\stot$ Not-Shortest and Next-to-Shortest Paths]
Let $G=(V,E,w)$ be a directed graph and let $s,t\in V$. 
For an $s\utov t$ path $P$, we define:
\begin{itemize}
    \item $P$ is an $\stot$ \textbf{not-shortest path} if $P$ is not a shortest path (i.e., $w(P) > d_G(s,t)$).
    \item $P$ is an $\stot$ \textbf{next-to-shortest path} if $P$ has the minimum length among all $s\utov t$ not-shortest paths. In other words, for any $s\utov t$ path $P_1$, if $w(P_1)<w(P)$, then $w(P_1)=d_{G}(s,t)$.
\end{itemize}
\end{definition}
\begin{definition}[The Next-to-Shortest Path Problem]
In the next-to-shortest path problem, we are given a directed graph $G=(V,E,w)$ and two vertices $s,t\in V$. The objective is to find a $s\utov t$ next-to-shortest path or to return that there is no such path. 
\end{definition}

If we only consider the decision version of the next-to-shortest path problem, we can notice that there exists a next-to-shortest path if and only if there exists a not-shortest path.

To show that a walk \( W \) is a not-shortest path, it is typically necessary to verify the following two conditions:
\begin{enumerate}
  \item \( W \) contains no repeated vertices; that is, \( W \) is a simple path.
  \item The cost of \( W \) is strictly greater than the length of a $\stot$ shortest path, i.e., \( w(W) > d_G(s, t) \).
\end{enumerate}

\paragraph{Back-Edges}
Given a graph $G$ with vertices $s,t$, 
we call an edge $(u,v)\in E$ a \textbf{back-edge} if $d_{G}(s,u)+w(u,v)>d_{G}(s,v)$. Otherwise, if $d_{G}(s,u)+w(u,v)=d_{G}(s,v)$, we call $(u,v)$ a \textbf{forward-edge}. In addition, we denote $E_{B}(G)$ and $E_{F}(G)$ as the sets of back-edge and forward-edge respectively.

Similar concepts have been studied in several prior works, including \cite{DBLP:journals/networks/WuW15}. We use the terminology \emph{back-edge} because when restricting to the class of \emph{$(s,t)$-layered} graphs that we will define in \Cref{def:layered-graph}, back-edges go backwards through the layers. 

\begin{lemma}\label{lem:basic-back-edges}
In a graph $G=(V,E,w)$, an $s \utov t$ path $P$ is an $s \utov t$ not-shortest path if and only if $\ESet(P)$ contains at least one back-edge.
\end{lemma}
\begin{proof}
Let $s\utov t$ path $P=(v_1,v_2,\ldots,v_k)$, where $k=|\VSet(P)|$. We know 
$$\begin{aligned}
    w(P) &=\sum_{i=1}^{k-1} w(v_i,v_{i+1})=\sum_{i=1}^{k-1} \left( d_{G}(s,v_{i+1})-d_{G}(s,v_{i})+\left(w(v_{i},v_{i+1})-d_{G}(s,v_{i+1})+d_{G}(s,v_i)\right)\right) \\
    &=d_{G}(s,t)+\sum_{i=1}^{k-1}\left(w(v_{i},v_{i+1})-d_{G}(s,v_{i+1})+d_{G}(s,v_i)\right).
\end{aligned}$$

Due to the definition of $d_{G}(\cdot,\cdot)$, we know $w(v_{i},v_{i+1})-d_{G}(s,v_{i+1})+d_{G}(s,v_i)\ge 0$ for all $1\le i<k$. Hence, $w(P)>d_{G}(s,t)$ if and only if there exists some $1\le j<k$, $w(v_{j},v_{j+1})-d_{G}(s,v_{j+1})+d_{G}(s,v_j)>0$, which means that $(v_{j},v_{j+1})$ is a back-edge.
\end{proof}
\paragraph{Vertex Disjoint Paths} Recall that the \( k \)-Vertex Disjoint Paths ($k$-VDP) problem asks for \( k \) vertex-disjoint paths in a graph, each connecting a specified pair of terminals.

If we only consider the decision version of next-to-shortest path (i.e., determine if there exists a next-to-shortest path), an intuitive approach is to enumerate all back-edges $(u,v)\in E$ and check if there exist 2 vertex-disjoint paths, one from $s$ to $u$ and the other from $v$ to $t$. Unfortunately, 2-VDP is NP-hard by \cite{FORTUNE1980111}.

For the $2$-VDP problem in a directed acyclic graph, we have the following lemma: 

\begin{lemma}[\cite{tholey2012linear}] \label{lem:DAG-VDP}
Given a directed acyclic graph $G$ with $n$ vertices and $m$ edges, and two vertex pairs $(s_1, t_1)$ and $(s_2, t_2)$, there exists an algorithm that, in $O(n + m)$ time, outputs vertex-disjoint paths from $s_1$ to $t_1$ and from $s_2$ to $t_2$, if such a pair of paths exists.
\end{lemma}

%% file: reduction.tex
\section{Reduction to (s,t)-Layered Graphs}\label{sec:reduction}

In this section, we reduce the next-to-shortest path problem on directed graphs with positive edge weights to a special class of digraphs which we called \textbf{$(s,t)$-layered} graphs. 

\begin{definition}[$(s,t)$-Layered Graph] \label{def:layered-graph}
Let $G=(V,E,w)$ be a positively weighted directed graph, and let $s,t\in V$. We say that $G$ is an \textbf{$(s,t)$-layered graph} if the following three conditions hold:
\begin{enumerate}
    \item For every vertex $u\in V$, $d_{G}(s,u)+d_{G}(u,t)=d_{G}(s,t)$. In other words, every vertex lies on at least one shortest $s\to t$ path.
    \item For every edge $(u,v)\in E$, $d_{G}(s,u)\not =d_{G}(s,v)$.
    \item For every edge $(u,v)\in E$ with $d_{G}(s,u) < d_{G}(s,v)$, we have that for all $x\in V$,
    \[
        d_{G}(s,x)\le d_{G}(s,u) \quad \text{or} \quad d_{G}(s,x)\ge d_{G}(s,u)+w(u,v).
    \]
\end{enumerate}
\end{definition}

Here we provide some intuition behind \Cref{def:layered-graph}. 
In an $(s,t)$-layered graph $G$, the vertex set $V$ admits a natural partition into layers according to the distance from $s$. Every forward edge goes only from one layer to the next, but never skips layers. In addition, each back-edge goes strictly backwards through the layers. (These structural properties are not guaranteed in a general directed graph. For instance, if we take an arbitrary directed graph and partition the vertices into layers by their distance from $s$, then the back-edges may not go backwards across the layers.) For more details about the intuition, see~\Cref{lem:layered-lambda} in~\Cref{sec:algorithm}.
Throughout this paper, when we refer to an $(s,t)$-layered (or $(s,t)$-straight, defined later) graph, it is implicitly understood that the next-to-shortest path problem under consideration is the problem of finding an $s\utov t$ next-to-shortest path (not merely deciding its existence). In order to analyze the complexity, we further introduce the following concepts:
\begin{itemize}
\item $V_{B}(G)=\lbrace u\in V:\exists v\in V, \,\, (u,v)\in E_B(G)\text{ or } (v,u)\in E_B(G)\rbrace$ is the set of vertices incident to back-edges.
\item $E_{FP}(G)=\lbrace ((u_1,v_1),(u_2,v_2))\in E_F(G)\times E_F(G): d_{G}(s,u_1)=d_{G}(s,u_2)\rbrace$ is the set of pairs of same-layer forward-edges. 
\end{itemize}
Here is the main theorem in this section:

\begin{theorem}\label{thm:reduction}
Assume that:
\begin{itemize}
\item There exists a $O\left(|V_{B}(G_L)|^2|E_{FP}(G_L)|\left(|E(G_L)|+|V(G_L)|\log |V(G_L)|\right)\right)$-time algorithm that can solve the next-to-shortest path problem on any $(s,t)$-layered graph $G_L$.
\end{itemize}
Then 
\begin{itemize}
\item There exists a $O(|V|^4|E|^3\log |V|)$-time algorithm that can solve the next-to-shortest path problem on any positively weighted directed graph $G=(V,E,w)$.
\end{itemize}
\end{theorem}


To prove this theorem, we proceed in two steps. First, we reduce the original problem to the next-to-shortest path problem restricted to graphs that satisfy Property~(1) of \Cref{def:layered-graph}, namely that every vertex lies on some shortest $\stot$ path; we refer to such graphs as \textbf{$(s,t)$-straight graphs}. Second, we further reduce the next-to-shortest path problem on $(s,t)$-straight graphs to the case of $(s,t)$-layered graphs. 

\subsection{Reducing to \texorpdfstring{$(s,t)$}{(s,t)}-Straight Graphs}
In this subsection, we reduce the original problem to the next-to-shortest path problem restricted to \textbf{$(s,t)$-straight graphs}.
The concept of $(s,t)$-straight graphs follows the terminology introduced by Berger, Seymour, and Spirkl~\cite{berger2021finding}. 
In that work, the authors studied induced paths in unweighted undirected graphs and showed that the problem of finding an induced not-shortest $(s,t)$-path can be reduced to the same problem on (unweighted, undirected) $(s,t)$-straight graphs. 
In this section, we apply a similar, though more involved, argument to reduce the next-to-shortest path problem (on positively weighted, directed graphs) to the same problem on (positively weighted, directed) $(s,t)$-straight graphs.

In this subsection, we are given a positively weighted directed graph $G$ and $s,t\in V$. To show the reduction, assume we have a polynomial-time algorithm $\nsps(G;s,t)$ that can solve the next-to-shortest path problem when $G$ is an $(s,t)$-straight graph. Based on such an oracle, we present a recursive algorithm $\nsp(G;s,t)$  that solves the next-to-shortest path problem when $G$ is any positively weighted directed graph. In particular, $\nsp(G;s,t)$ (or $\nsps(G;s,t)$) outputs a next-to-shortest path if one exists, or $\bot$ otherwise. 
In the second case, where there is no next-to-shortest path, we define the weight $w$ of the output as $w(\bot)=+\infty$. 

If $G$ is already an $(s,t)$-straight graph, then we call $\nsps(G;s,t)$. Otherwise, we can find some $u\in V$ that $d_{G}(s,u)+d_{G}(u,t)>d_{G}(s,t)$. 
If $d_{G}(s,u)=\infty$ or $d_{G}(u,t)=\infty$, we can simply output $\nsp(G\lbrack V\setminus \lbrace u\rbrace\rbrack;s,t)$ since $u$ doesn't appear in any $s\utov t$ path. 

Therefore, we can assume that $d_{G}(s,t)<d_{G}(s,u)+d_{G}(u,t)<\infty$. 
Let $N^{in}(u)$ and $N^{out}(u)$ be the in-neighbors and the out-neighbors of $u$, respectively, i.e., 
$$\begin{aligned}
N^{in}(u)=\left\lbrace v\in V: (v,u)\in E\right\rbrace,&&
N^{out}(u)=\left\lbrace v\in V:(u,v)\in E\right\rbrace.\end{aligned}$$

We would like to construct a graph $G'=(V',E',w')$ that does not contain $u$, since $u$ violates the definition of $(s,t)$-straight. To do so, we will replace $u$ with the set of all edges from its in-neighborhood to its out-neighborhood with appropriate weights (omitting self-loops), as shown in \Cref{fig:reduction}.
Formally, denote $N^{\times}(u)=(N^{in}(u)\times N^{out}(u))\setminus \lbrace (x,x):x\in V\rbrace$.
\begin{definition}\label{def:G'}
We construct $G'$ from $G$ by removing vertex $u$ and adding edges connecting the neighbors $N^{in}(u)$ and $N^{out}(u)$. In particular, $V'=V\setminus \lbrace u\rbrace$, $E'=\lbrace (x,y)\in E:x\not =u, y\not =u\rbrace\cup N^{\times}(u)$, and for $(x,y)\in E'$, we set

$$w'(x,y)=\left\lbrace \begin{aligned}
&w(x,y) && \text{if} && (x,y)\in E \text{ and } (x,y)\not \in N^{\times}(u);\\
&w(x,u)+w(u,y) && \text{if} && (x,y)\not\in E \text{ and } (x,y)\in N^{\times}(u);\\
&\min(w(x,y),w(x,u)+w(u,y)) && \text{if} && (x,y)\in E \text{ and } (x,y) \in N^{\times}(u).
\end{aligned}\right.$$
\end{definition}

\begin{figure}[h!]
    \centering
    \begin{subfigure}[t]{0.45\linewidth}
        \centering
        \includegraphics[width=\linewidth,
            trim=2cm 2cm 3cm 1.2cm, clip]{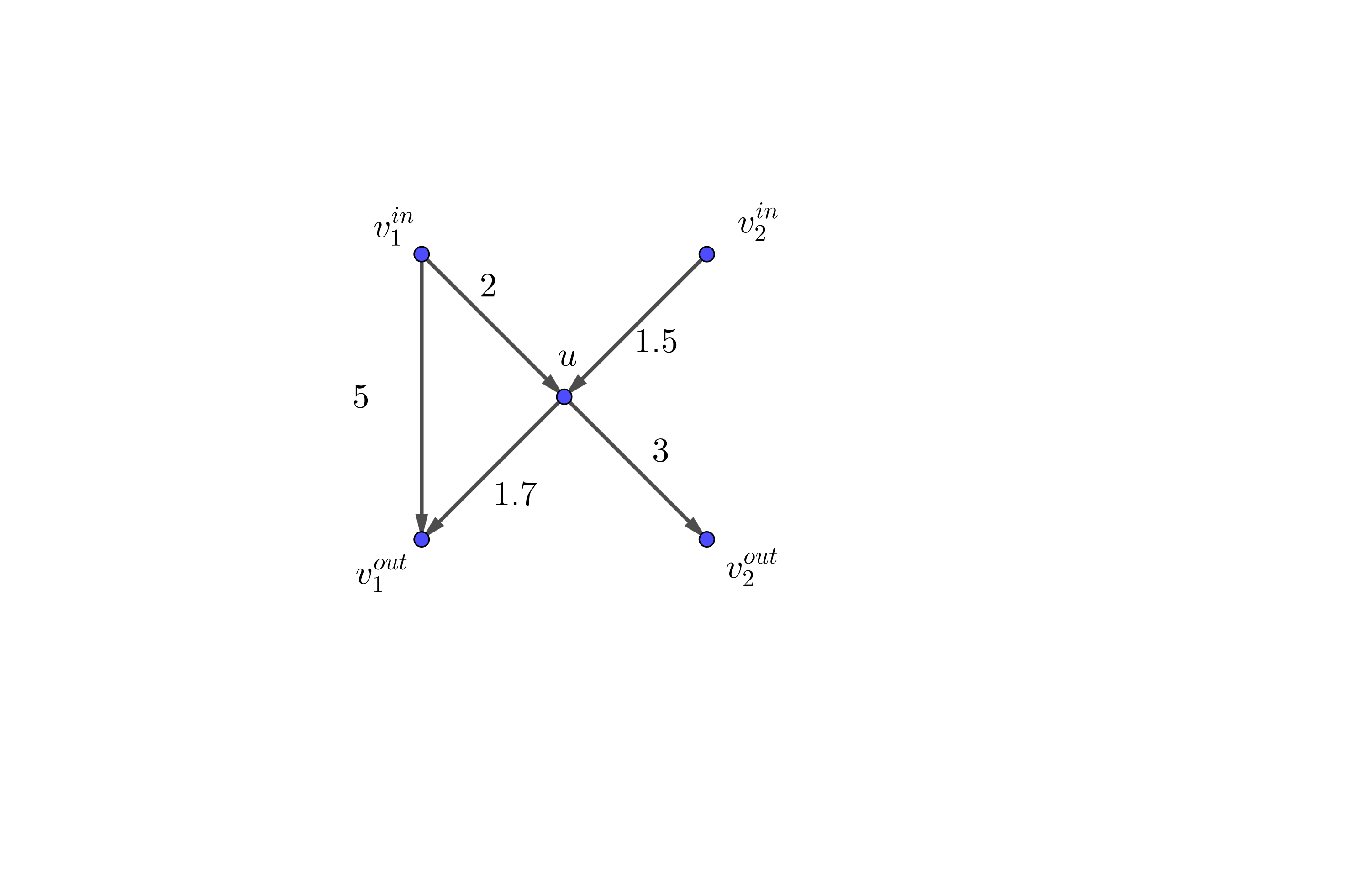}
        \caption{A subgraph of $G$.}
        \label{fig:reduction-1}
    \end{subfigure}%
    \hfill
    \begin{subfigure}[t]{0.45\linewidth}
        \centering
        \includegraphics[width=\linewidth,
            trim=2cm 2cm 3cm 1.2cm, clip]{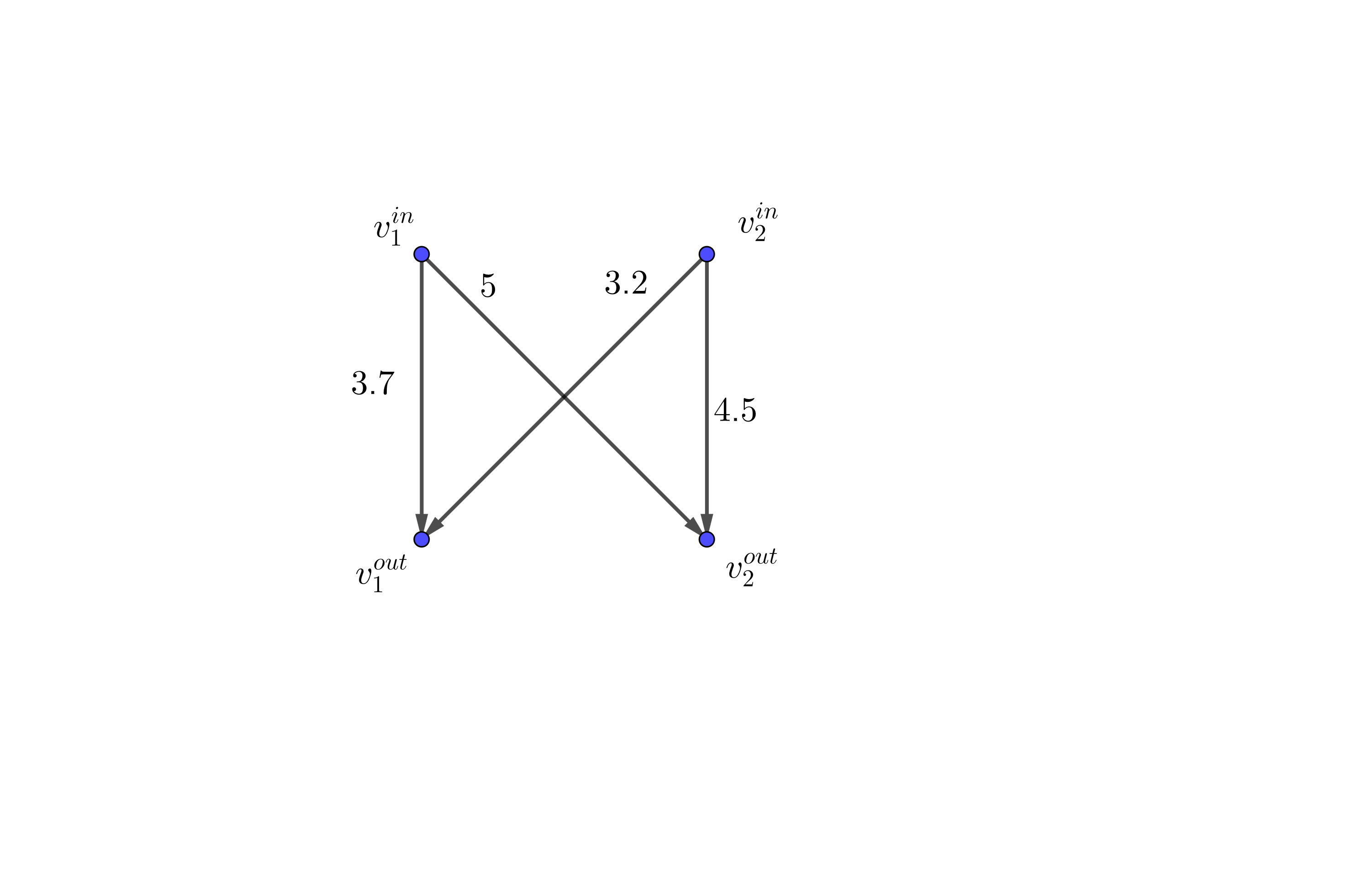}
        \caption{The modified subgraph in $G'$.}
        \label{fig:reduction-2}
    \end{subfigure}
    \caption{Illustration of the reduction: (a) a subgraph of $G$, and (b) the corresponding subgraph in $G'$ obtained by modification.}
    \label{fig:reduction}
\end{figure}

\begin{definition}[Shortcut Edge] \label{def:shortcut-edge}
Call an edge $(v^{in},v^{out})\in N^{\times}(u)$ a \textbf{shortcut edge} if 
\[
(v^{in},v^{out})\notin E \ \text{or}\ w'(v^{in},v^{out})<w(v^{in},v^{out}).
\]
\end{definition}


There is a natural correspondence between paths in $G$ and paths in $G'$, described in the following two definitions: 
\begin{definition}[$\pi$]
For $x,y\in V'$ and an $x\utov y$ path $P$ in $G$ (or $P=\bot$), define $\pi(P)$ as follows:
\begin{itemize}
    \item If $P$ does not contain $u$, then $\pi(P) = P$.
    \item If $P$ contains $u$, let $v^{in}$ and $v^{out}$ denote the predecessor and successor of $u$ on $P$, respectively. 
    Then $\pi(P)$ is the path obtained by replacing the subpath $(v^{in},u,v^{out})$ with the single edge $(v^{in},v^{out})$. 
\end{itemize}
\end{definition}


\begin{definition}[$\pi'$] \label{def:pi-prime}
For $x,y\in V'$ and an $x\utov y$ path $P'$ in $G'$ (or $P'=\bot$), define $\pi'(P')$ as follows: 
\begin{itemize}
\item If $P'$ contains no shortcut edges, define $\pi'(P')=P'$. 
\item If $P'$ contains at least one shortcut edge, let
$(v^{in}_1,v^{out}_1)$ and $(v^{in}_2,v^{out}_2)$ be the first and last shortcut edges
traversed by $P'$ (in order along $P'$). Define
\[
\pi'(P') := P'_{x\to v^{in}_1} \circ R \circ P'_{v^{out}_2 \to y},
\]
where $R=(v^{in}_1,u,v^{out}_2)$.
\ignore{
\[
R \;=\;(v^{in}_1,u,v^{out}_2).
\]
Writing $P'_{a\to b}$ for the subpath of $P'$ from $a$ to $b$, and using $\circ$ for path concatenation, define
\[
\pi'(P') \;=\; P'_{x\to v^{in}_1} \;\circ\; R \;\circ\; P'_{v^{out}_2 \to y}.
\]}
\end{itemize}
\end{definition}

Now we will prove a lemma relating the paths and distances in $G$ and $G'$:
\begin{lemma}\label{lem:map-paths}
For any $x,y\in V'$, we have the following properties:
\begin{enumerate}[label=(\alph*)]
\item For any $x\utov y$ path $P$ in $G$, $\pi(P)$ is a $x\utov y$ path in $G'$ and $w'(\pi(P))\le w(P)$.
\item For any $x\utov y$ path $P'$ in $G'$, $\pi'(P')$ is a $x\utov y$ path in $G$ and $w(\pi'(P'))\le w'(P')$.
\item $d_{G}(x,y)=d_{G'}(x,y)$.
\end{enumerate}
\end{lemma}
\begin{proof}
\mbox{}\\
\noindent \emph{Proof of (a).} Since $u\not \in V'$, the mapping $\pi(P)$ does not change the endpoints of $P$. 
By~\Cref{def:G'}, $w'(a,b)\le w(a,b)$ for every edge $(a,b)\in E$ with $a,b\in V'$. If $P$ does not go through $u$, then $w'(\pi(P))\le w(P)$. Otherwise, let $v^{in}$ and $v^{out}$ denote, respectively, the predecessor and successor of $u$ on $P$. We know that $\pi(P)$ is obtained from $P$ by removing $u$ and adding $(v^{in},v^{out})$.  Since $w'(v^{in},v^{out})\le w(v^{in},u)+w(u,v^{out})$, $w'(\pi(P))\le w(P)$.

\noindent\emph{Proof of (b).} It is clear that $\pi'(P')$ is also a $x\utov y$ path in $G$. Assume that $P'$ goes through $k$ shortcut edges $(v^{in}_1,v^{out}_1),\ldots,(v^{in}_k,v^{out}_k)$. 
\begin{itemize}
    \item If $k=0$, then it is clear that $w(\pi'(P'))=w'(P')$.
    \item If $k\ge 1$, since $(v^{in}_i,v^{out}_i)$ are shortcut edges (\Cref{def:shortcut-edge}) for $1\le i\le k$, we have
    \[
    w'(v^{in}_i,v^{out}_i)=w(v^{in}_i,u)+w(u,v^{out}_i).
    \]
    Hence
    \[\begin{aligned}
    w'\!\left(P'_{v^{in}_1\to v^{out}_k}\right)&=\sum_{1\le i\le k} w'(v^{in}_i,v^{out}_i)
    = \sum_{1\le i\le k} \left(w(v^{in}_i,u)+w(u,v^{out}_i)\right)\\
    &\ge w(v^{in}_1,u)+w(u,v^{out}_k).
    \end{aligned}
    \]
    Therefore $w(\pi'(P'))\le w'(P')$.
\end{itemize}

\noindent \emph{Proof of (c).} We consider the shortest path from $x$ to $y$ ($x,y\in V'$). (a) implies $d_{G'}(x,y)\le d_{G}(x,y)$, and (b) implies $d_{G}(x,y)\le d_{G'}(x,y)$. Hence $d_{G'}(x,y)=d_{G}(x,y)$.

\end{proof}

\begin{lemma}\label{lem:no-new-forward}
For any forward-edge $(x,y)\in E(G')$,  if $y$ lies on some shortest $s\utov t$ path in $G'$, then $(x,y)$ is also a forward-edge in $G$. That is,
\[\begin{aligned}
&\{(x,y)\in E_F(G') : d_{G'}(s,y)+d_{G'}(y,t)=d_{G'}(s,t)\}\\
&\subseteq
\{(x,y)\in E_F(G) : d_G(s,y)+d_G(y,t)=d_G(s,t)\}.
\end{aligned}
\]
\end{lemma}

\begin{proof}
Consider any $(x,y)\in E_F(G')$ such that
\[
d_{G'}(s,y)+d_{G'}(y,t)=d_{G'}(s,t).
\]
By~\Cref{lem:map-paths}(c), $d_G(s,x)=d_{G'}(s,x)$ and $d_G(s,y)=d_{G'}(s,y)$.

\begin{itemize}

\item If $(x,y)$ is not a shortcut edge, then $w'(x,y)=w(x,y)$, and hence
\[
d_G(s,x)+w(x,y)=d_G(s,y).
\]
Thus $(x,y)\in E_F(G)$.

\item Otherwise, $(x,y)$ is a shortcut edge and
\[
w'(x,y)=w(x,u)+w(u,y).
\]
Hence
\[
d_{G'}(s,x)+w'(x,y)
= d_G(s,x)+w(x,u)+w(u,y)
\ge d_G(s,u)+w(u,y).
\]

Since $d_G(s,u)+d_G(u,t)>d_G(s,t)$ and
$d_G(s,y)+d_G(y,t)=d_G(s,t)$,
it follows that
\[
d_{G'}(s,x)+w'(x,y) > d_{G'}(s,y).
\]

This contradicts the assumption that $(x,y)\in E_F(G')$.
\end{itemize}
\end{proof}

Given an oracle $\nsps(G;s,t)$, the algorithm is formally defined in the algorithm block below. The idea of this algorithm is to recursively call $\nsp(G';s,t)$. However, it would not be accurate to directly return the output. The reason is that $u$ could appear in the only next-to-shortest path $P$ in $G$, but the corresponding path $\pi(P)$ in $G'$ may be a \emph{shortest} path, since the weight of an edge $(x,y)\in E\cap N^{\times}(u)$ is defined as $\min(w(x,y),w(x,u)+w(u,y))$. 
For this reason, in such cases, we explicitly compute these candidate next-to-shortest paths that are not preserved in $G'$, in step~\ref{step:nextsp} of the algorithm block below.

\begin{mdframed}
\vspace{-1em}
\paragraph{Algorithm $\nsp(G;s,t)$}\mbox{}\\
\textit{Input:} A positively weighted directed graph $G=(V,E,w)$, and two vertices $s,t\in V$ \\
\textit{Output:} A $s\utov t$ next-to-shortest path if one exists, or report $\bot$ otherwise.
\begin{enumerate}
\item If $G$ is an $(s,t)$-straight graph, then return $\nsps(G;s,t)$. \label{step:boundary-1}
\item Select a vertex $u\in V$ such that $d_{G}(s,u)+d_{G}(u,t)\not =d_{G}(s,t)$. 
\item If $d_{G}(s,u)=\infty$ or $d_{G}(u,t)=\infty$, then return $\nsp(G\lbrack V\setminus \lbrace u\rbrace\rbrack;s,t)$. \label{step:boundary-2}
\item Construct $G'$ based on $u$ and initialize $\hat{P} \gets \pi'(\nsp(G';s,t))$. 
\item For each edge $(x,y)\in E\cap N^{\times}(u)$ with $w(x,y)<w(x,u)+w(u,y)$, do: \label{step:nextsp}
\begin{enumerate}
    \item If there exists an $s\utov t$ shortest path in $G$ that traverses $(x,y)$, 
          construct a path $Q$ by replacing $(x,y)$ with $(x,u),(u,y)$. 
          (Since $u$ cannot appear on any $s \utov t$ shortest path, 
          the resulting path $Q$ is not a shortest path.)
    \item Otherwise, set $Q \gets \bot$.
    \item If $w(Q)<w(\hat{P})$, then update $\hat{P} \gets Q$. 
\end{enumerate}
\item Return $\hat{P}$.
\end{enumerate}
\end{mdframed}

\begin{lemma}\label{lem:G-to-G'}
We have the following properties: 
\begin{enumerate}[label=(\alph*)]
\item If there exists a next-to-shortest path $P$ in $G$, then either $d_{G'}(s,t)<w'(\pi(P))=w(P)$, or $\nsp(G;s,t)$ 
finds a next-to-shortest path  
in step~\ref{step:nextsp} of the algorithm.
\item If there exists a next-to-shortest path $P'$ in $G'$, then $d_{G}(s,t)<w(\pi'(P'))\le w'(P')$.
\end{enumerate}
\end{lemma}
\begin{proof}$ $

\noindent \emph{Proof of (a).}
Assume that there is a next-to-shortest path $P$ in graph $G$. By~\Cref{lem:map-paths} (c), we know $d_{G'}(s,t)=d_{G}(s,t)$. In addition, according to the definition of $\pi$, if $P$ does not go through $u$, $w'(\pi(P))=w(P)$ and $\pi(P)$ is also a not-shortest path in $G'$ since $d_{G'}(s,t)=d_{G}(s,t)$. Therefore, we can assume that $P$ goes through $u$, and $v^{in}$ and $v^{out}$ denote, respectively, the predecessor and successor of $u$ on $P$.
\begin{itemize}
    \item If $(v^{in},v^{out})\not \in E$ or $w(v^{in},v^{out})\ge w(v^{in},u)+w(u,v^{out})$, then $w'(\pi(P))=w(P)$, and $\pi(P)$ is also a not-shortest path in $G'$. 
    \item If $(v^{in},v^{out})\in E$, $w(v^{in},v^{out})< w(v^{in},u)+w(u,v^{out})$, and there is a $s\utov t$ shortest path in $G$ that goes through $(v^{in},v^{out})$, then we check the edge $(v^{in},v^{out})$ in step~\ref{step:nextsp} of the algorithm $\nsp(G;s,t)$.  Thus, we find a path of length $d_{G}(s,v^{in})+w(v^{in},u)+w(u,v^{out})+d_{G}(v^{out},t)$. Because the next-to-shortest path $P$ goes through $u$, $v^{in}$, and $v^{out}$, we know that $w(P)\ge d_{G}(s,v^{in})+w(v^{in},u)+w(u,v^{out})+d_{G}(v^{out},t)$. Since our algorithm finds a path of that weight, our algorithm indeed finds a next-to-shortest path. 
    \item Otherwise, $(v^{in},v^{out})\in E$, 
    $w(v^{in},v^{out})< w(v^{in},u)+w(u,v^{out})$, and there is no $s\utov t$ shortest path in $G$ that goes through $(v^{in},v^{out})$. By replacing $P$'s subpath $(v^{in},u,v^{out})$ by the edge $(v^{in},v^{out})$, we obtain a not-shortest path with smaller length, contradicting the fact that $P$ is a next-to-shortest path.  
\end{itemize}

\noindent \emph{Proof of (b).} Let $P'$ be a next-to-shortest path in $G'$. By~\Cref{lem:map-paths} (b), we know $w'(P')\ge w(\pi'(P'))$. To show $w(\pi'(P'))>d_{G}(s,t)$, we consider two cases:
\begin{itemize}
\item If $P'$ does not contain any shortcut edges, then $w(\pi'(P'))= w'(P')>d_{G'}(s,t)=d_{G}(s,t)$.
\item If $P'$ has
at least one shortcut edge, then $\pi'(P')$ goes through $u$, which means $w(\pi'(P'))>d_{G}(s,t)$. 
\end{itemize}
\end{proof}

Observe that during the process of reducing a graph $G$ to an $(s,t)$-straight graph $G_S$, the weight of every edge $(a,b)$ is always equal to the length of some $a\utov b$ path in the original graph, so the edge weights are bounded.

\begin{lemma}\label{lem:reduce-to-straight}
Assume there exists an  algorithm that solves the next-to-shortest path problem on any $(s,t)$-straight graph $G_S$ in:
$$O(|V(G_S)|^3\cdot |E_F(G_S)|^2\cdot \left(|V(G_S)||E_F(G_S)|+|E(G_S)|\right)\cdot \log |V(G_S)|)$$
time, then there exists a $O(|V|^4|E|^3\log |V|)$-time algorithm that can solve the next-to-shortest path problem on any positively weighted digraph $G=(V,E,w)$. 
\end{lemma}
\begin{proof}
We consider the above algorithm $\nsp(G;s,t)$, where $G$ is a positively weighted directed graph, and this recursive algorithm eventually call $G_S$.
Regarding the time complexity, we can notice the followings:
\begin{itemize}
    \item In each recursive step the vertex set decreases by one, i.e., $|V'| = |V|-1$. Hence, $\nsp(G;s,t)$ performs at most $O(|V|)$ recursive calls before $G$ becomes $(s,t)$-straight.
    \item By~\Cref{lem:no-new-forward}, $$\begin{aligned}
        E_F(G_S)&=\lbrace (x,y)\in E_{F}(G_S):d_{G_S}(s,y)+d_{G_S}(y,t)=d_{G_S}(s,t)\rbrace\\ 
        &\subseteq \lbrace (x,y)\in E_{F}(G):d_{G}(s,y)+d_{G}(y,t)=d_{G}(s,t)\rbrace\\
        &\subseteq E_F(G)\subseteq E(G).
        \end{aligned}$$
        Hence, $|E_F(G_S)|\le |E(G)|$
    \item Each step takes $O(|V|^2 |E_F(G)|)$ time, since in each step we examine at most $|V|^2$ pairs of vertices, and for each pair we spend $O(1)$ time plus an additional scan over a subset of forward-edges which lie on some shortest $s\utov t$ paths.
\end{itemize}
Consequently, if $\nsps(G_S;s,t)$ runs in $$O(|V(G_S)|^3\cdot |E_F(G_S)|^2\cdot \left(|V(G_S)||E_F(G_S)|+|E(G_S)|\right)\cdot \log |V(G_S)|)$$ time, then the time complexity of $\nsp(G;s,t)$ is:

$$\begin{aligned}&O(|V(G_S)|^3\cdot |E_F(G_S)|^2\cdot \left(|V(G_S)||E_F(G_S)|+|E(G_S)|\right)\cdot \log |V(G_S)|)+O(|V|)\cdot O(|V|^2|E_{F}(G)|)\\
&\quad\quad\quad \le O(|V|^3|E_F(G)|^2\cdot (|V||E_F(G)|+|V|^2)\cdot \log |V|)+O(|V|)\cdot O(|V|^2E_F(G)|)\\
&\quad\quad\quad \le O(|V|^4|E|^3\log |V|).\end{aligned}$$

With regard to the correctness, since the two boundary cases (step~\ref{step:boundary-1} and step~\ref{step:boundary-2} in $\nsp(G;s,t)$) are clear, we only need to discuss the case where $G$ is not $(s,t)$-straight, and we have a vertex $u\in V$  such that $d_{G}(s,t)<d_{G}(s,u)+d_{G}(u,t)<\infty$. In this case, we construct another graph $G'=(V',E',w')$ based on $u$, and let the final result of $\nsp(G;s,t)$ be $\hat{P}$. 

\begin{itemize}
    \item If there exists a next-to-shortest path $P$ in $G$, then according to~\Cref{lem:G-to-G'} (b), either $\nsp(G';s,t)=\bot$ or $w(\pi'(\nsp(G';s,t)))\ge w(P)$, so $w(\hat{P})\ge w(P)$. Conversely, according to~\Cref{lem:G-to-G'} (a), $w(\hat{P})\le w(P)$. Hence, $w(\hat{P})=w(P)$. 
    \item If there is no next-to-shortest path in $G$, according to~\Cref{lem:G-to-G'} (b), $\nsp(G';s,t)=\bot$. Moreover, step~\ref{step:nextsp} of $\nsp(G;s,t)$ will not find a next-to-shortest path since it only explores paths in $G$. So $\hat{P}=\bot$. 
\end{itemize}
\end{proof}

\subsection{From (s,t)-Straight Graphs to (s,t)-Layered Graphs}\label{subsec:striaght-layered}

In this subsection, we reduce the next-to-shortest path problem on $(s,t)$-straight graphs to the problem on $(s,t)$-layered graphs. In order to measure how close an $(s,t)$-straight graph is to an $(s,t)$-layered graph, we define a potential function $\varphi(G)$ ($G$ is an $(s,t)$-straight graph) to estimate how many edges violate the properties of $(s,t)$-layered graphs. In particular, for an $(s,t)$-straight graph $G$, we define: 

$$\begin{aligned}
    \varphi(G) &=\#\lbrace (u,v)\in E \mid d_{G}(s,u)=d_{G}(s,v)\rbrace\\
    &+\#\lbrace (u,v)\in E \mid 
d_{G}(s,u)<d_{G}(s,v) \text{ and } \exists x\in V, d_{G}(s,u)< d_{G}(s,x)< d_{G}(s,u)+w(u,v)\rbrace.
\end{aligned}$$
   
By~\Cref{def:layered-graph}, a graph $G$ is $(s,t)$-layered if and only if $G$ is $(s,t)$-straight and $\varphi(G)=0$. Given an algorithm $\nspl(G;s,t)$ that can solve the next-to-shortest path problem when $G$ is an $(s,t)$-layered graph, we present a recursive algorithm $\nsps$: 

Basically, we consider an edge that violates the properties of the $(s,t)$-layered graph. If it is a back-edge, we remove it. If it is a forward-edge we subdivide it so that it doesn't violate the property any more.\\$ $\\

\begin{mdframed}
\vspace{-1em}
\paragraph{Algorithm $\nsps(G;s,t)$}\mbox{}\\
\textit{Input:} An $(s,t)$-straight graph $G=(V,E,w)$, and two vertices $s,t\in V$ \\
\textit{Output:} A $s\utov t$ next-to-shortest path if one exists, or report that none exists.
\begin{enumerate}
\item If $G$ is an $(s,t)$-layered graph, then we return $\nspl(G;s,t)$. Otherwise, We select an edge $(u,v)\in E$ such that either $d_{G}(s,u)=d_{G}(s,v)$, 
or both $d_{G}(s,u)<d_{G}(s,v)$ and there exists a vertex $x$ with $d_{G}(s,u)<d_{G}(s,x)<d_{G}(s,u)+w(u,v).$  \label{step:nsps-1}

\item If $(u,v)$ is a back-edge, 
we construct a new graph $G'$ from $G$ by removing the edge $(u,v)$. 
Since $G$ is an $(s,t)$-straight graph, there exists a shortest $s\utov u$ path $\anyP_{s\to u}$ 
and a shortest $v\utov t$ path $\anyP_{v\to t}$. 
We then compare the weight of the concatenated path $\anyP_{s\to u}\circ (u,v)\circ \anyP_{v\to t}$ 
with $\nsps(G';s,t)$, and return the one of smaller weight. \label{step:nsps-2}
\item If $(u,v)$ is a forward-edge (i.e., $d_{G}(s,u)+w(u,v)=d_{G}(s,v)$), we denote $$\lbrace q_1,q_2,\ldots,q_k\rbrace=\lbrace d_{G}(s,x): x\in V, d_{G}(s,u)<d_{G}(s,x)<d_{G}(s,u)+w(u,v)\rbrace,$$ where $k$ is the number of such values, and $q_1<q_2<\ldots <q_k$.  Then we construct a new graph $G'=(V',E',w')$ by replacing $(u,v)$ by a list of $k+2$ vertices. In particular, we create a series of new vertices $v_1,v_2,\ldots,v_k$. To simplify notation, we let $v_0=u$, $v_{k+1}=v$, $q_0=d_{G}(s,u)$, and $q_{k+1}=d_{G}(s,v)$). Then, we define $
V'=V\cup \lbrace v_1,v_2,\ldots,v_k\rbrace$ and 
$E'=(E\setminus \lbrace (u,v)\rbrace)\cup \lbrace (v_i,v_{i+1}):0\le i\le k\rbrace$. 
For the edge weights, we define $w'(v_i,v_{i+1})=q_{i+1}-q_i$ for $1\le i\le k$, and for all  edges $e\in E\setminus \lbrace (u,v)\rbrace$, we set $w'(e)=w(e)$. Then we return $\nsps(G';s,t)$. \label{step:nsps-3}
\end{enumerate}
\end{mdframed}

\begin{lemma}\label{lem:reduction:dec-potential}
Whenever the algorithm $\nsps(G;s,t)$ makes a recursive call on a graph 
$G'=(V',E',w')$, we have: 
\begin{enumerate}[label=(\alph*)]
\item For any $z\in V$, $d_{G}(s,z)=d_{G'}(s,z)$. 
\item $\#\lbrace d_{G}(s,z):z\in V\rbrace=\#\lbrace d_{G'}(s,z):z\in V'\rbrace$ 
\item $\varphi(G')= \varphi(G)-1$. 
\end{enumerate}
\end{lemma}
\begin{proof}
In the algorithm, we select an edge $(u,v)$ that violates the defining property of an $(s,t)$-layered graph.

\emph{Proof of (a).}
If $(u,v)$ is a back-edge, as in step~\ref{step:nsps-2}, for any $z\in V$, there is no $s\utov z$ shortest path that goes through  $(u,v)$, so $d_{G'}(s,z)=d_{G}(s,z)$.

If $(u,v)$ is a forward-edge, as in step~\ref{step:nsps-3}, then our algorithm replaces the forward-edge by a path of the same length. For any $z\in V$, any $s\utov z$ shortest path in $G'$ that goes through any new vertices must go through $u$ and $v$, and we can obtain a corresponding $s\utov z$ path in $G$ by replacing this $u\utov v$ subpath by an edge $(u,v)$. Therefore, $d_{G'}(s,z)=d_{G}(s,z)$. 

\emph{Proof of (b)}
According to (a), $d_{G'}(s,z)=d_{G}(s,z)$ for all $z\in V$, and for new vertices, it is clear that $\lbrace d_{G'}(s,z):z\in V'\setminus V\rbrace\subseteq \lbrace d_{G}(s,z):z\in V\rbrace$. 

\emph{Proof of (c)}
If $(u,v)$ is a back-edge, $(u,v)$ is removed in $G'$ and this process does not change other edges. Therefore, $\varphi(G')=\varphi(G)-1$.

If $(u,v)$ is a forward-edge, $(u,v)$ is removed in $G'$. With regard to the new edges, for every $0\le i\le k$, there are no vertices $z\in V$  satisfying $d_{G'}(s,v_i)<d_{G}(s,z)<d_{G'}(s,v_{i+1})=d_{G'}(s,v_i)+w'(v_i,v_{i+1})$. Therefore, $(v_i,v_{i+1})(0\le i\le k)$ does not violate the property but $(u,v)$ is removed, so $\varphi(G')=\varphi(G)-1$.
\end{proof}
\begin{lemma}\label{lem:straight-to-layered}
Assume that there exists a $O(|V_{B}(G_L)|^2\cdot |E_{FP}(G_L)|\cdot (|E(G_L)|+|V(G_L)|\log |V(G_L)|))$-time algorithm that can solve the next-to-shortest path problem on any $(s,t)$-layered graph $G_L$. 
There exists a $ O(|V(G_S)|^3\cdot |E_F(G_S)|^2\cdot \left(|V(G_S)||E_F(G_S)|+|E(G_S)|\right)\cdot \log |V(G_S)|)$-time algorithm that can solve the next-to-shortest path problem on any $(s,t)$-straight graph $G_S$. 
\end{lemma}
\begin{proof}
We consider the process $\nsps(G_S;s,t)$ that finally calls $\nspl(G_L;s,t)$.  
According to~\Cref{lem:reduction:dec-potential} (b), $\#\lbrace d_{G_S}(s,z):z\in V(G_S)\rbrace$ does not change during the recursive process, so for each recursive step, we add at most $O(\#\lbrace d_{G_S}(s,z):z\in V(G_S)\rbrace)$ new vertices. 

For an arbitrary recursive step with an input $G=(V,E)$ that recursively calls $G'$, we can notice that:
\begin{itemize}
\item If this step removes a back-edge, then $|V_{B}(G')|\le |V_{B}(G)|$ since the some vertices may no longer connect to back-edges. 
\item If this step replaces an edge by a list of vertices, then the new vertices are not connected to back-edges, so $|V_{B}(G')|\le|V_{B}(G)|$  
\end{itemize}

Therefore, $|V_{B}(G_L)|\le |V_{B}(G_S)|\le |V(G_S)|$

In addition, for every $q\in \lbrace d_{G}(s,z):z\in V\rbrace$, $\#\lbrace (u,v)\in E_F(G):d_{G}(s,u)=q\rbrace\le |E_F(G_S)|$. 
Hence, $|E_{FP}(G_L)|\le \#\lbrace d_{G_S}(s,z):z\in V(G_S)\rbrace \cdot |E_F(G_S)|^2$ 

By~\Cref{lem:reduction:dec-potential} (c), our algorithm $\nsps$ only repeats at most $\varphi(G_S)=O(|E(G_S)|)$ times, and each step takes $O(|E(G)|+|V(G)|)\le O(|E(G_S)|\cdot |V(G_S)|)$ time. Therefore, the total time complexity is:

$$\begin{aligned}O\big(|&V_{B}(G_L)|^2\cdot |E_{FP}(G_L)|\cdot (|E(G_L)|+|V(G_L)|\log |V(G_L)|)\big)\\
&\hspace{2cm}+O(|E(G_S)|\cdot |V(G_S)|)\cdot \#\lbrace d_{G_S}(s,z):z\in V(G_S)\rbrace\\
&\le O(|V(G_S)|^2\cdot |V(G_S)||E_F(G_S)|^2\cdot \left(|V(G_S)||E_F(G_S)|+|E(G_S)|\right)\cdot \log (|V(G_S)|\cdot |E(G_S)|))\\
&\hspace{2cm}+O(|E(G_S)|\cdot |V(G_S)|^2)\\
&\le O(|V(G_S)|^3\cdot |E_F(G_S)|^2\cdot \left(|V(G_S)||E_F(G_S)|+|E(G_S)|\right)\cdot \log |V(G_S)|).
\end{aligned}$$
To show the correctness, let $(u,v)$ be the edge selected in step~\ref{step:nsps-1} of $\nsps(G;s,t)$. Recall that $\anyP_{s\to u}$ and $\anyP_{v\to t}$ are arbitrary $s\utov u$ and $v\utov t$ shortest paths in an $(s,t)$-straight graph $G$, respectively. If $(u,v)$ is a back-edge and there is a next-to-shortest path $P$ that goes through $(u,v)$, we consider a walk $\anyP_{s\to u}\circ (u,v)\circ \anyP_{v\to t}$. Since $d_{G}(s,u)\le d_{G}(s,v)$ and along both $\anyP_{s\to u}$ and $\anyP_{v\to t}$, the distances from $s$ are monotonically increasing, this walk is a simple path. In addition, we know $\anyP_{s\to u}\circ (u,v)\circ \anyP_{v\to t}$ is the shortest among the paths that traverse $(u,v)$. Therefore, $\anyP_{s\to u}\circ (u,v)\circ \anyP_{v\to t}$ must be a next-to-shortest path. If $(u,v)$ is a forward-edge, then we have a bijection from $s\utov t$ paths in $G$ to $s\utov t$ paths in $G'$ by replacing $(u,v)$ with $(v_0,v_1,\ldots,v_{k+1})$ (and conversely).
\end{proof}
\begin{proofofthm}{\ref{thm:reduction}}
This statement follows from the combination of \Cref{lem:reduce-to-straight} and \Cref{lem:straight-to-layered}.
\end{proofofthm}

%% file: algorithm.tex
\section{Polynomial-time Algorithm for (s,t)-Layered Digraphs}\label{sec:algorithm}

In the following sections, we restrict our attention to $(s,t)$-layered graphs. In this section, we provide a polynomial-time algorithm $\nspl$ that can solve the next-to-shortest path problem on $(s,t)$-layered graphs. This algorithm is based on the Key Lemma (\Cref{lem:key}), which will be proved in~\Cref{sec:proof-of-key-lemma}. In the remainder of this paper, we fix an $(s,t)$-layered graph 
$G=(V,E,w)$. 
We first define the layer function for the $(s,t)$-layered graph as follows:
\begin{definition}[Layer Function]
\label{def:layer-function}
We are given an $(s,t)$-layered graph $G=(V,E,w)$. 
Let $k$ be the number of distinct values of $d(x)=d_{G}(s,x)$ over all $x \in V$, and write
\[
\{ d(x):x\in V \} = \{ p_1,p_2,\ldots,p_k\}, \quad
0=p_1<p_2<\cdots <p_k=d_{G}(s,t).
\]
We define the \textbf{layer function} $\lambda=\lambda_G: V \to [k]$ by setting 
$
\lambda(u)=i \quad\text{whenever } d(u)=p_i.
$
\end{definition}

Based on~\Cref{def:layered-graph} and~\Cref{def:layer-function}, 
we adopt the following notational conventions:

\begin{itemize}
    \item For $u\in V$, let $d(u)=d_{G}(s,u)$ denote the length of a shortest path from $s$ to $u$. 
    \item For a $u\utov v$ path $P$ ($u,v\in V$), we call $P$ \textbf{forward} if it only contains forward-edges. (Equivalently, along $P$ the layer function $\lambda$ is increasing.) 
    \item For each $u\in V$, since there exists a shortest $s\utov t$ path passing through $u$, there exist an $s\utov u$ forward path and a $u\utov t$ forward path. We denote by $\anyP_{s\to u}$ and $\anyP_{u\to t}$ arbitrary but fixed choices of such forward paths, respectively.  
\end{itemize}

The following lemma introduces some basic properties of $\lambda$:
\begin{lemma}\label{lem:layered-lambda}
Let $G=(V,E,w)$ be an $(s,t)$-layered graph with distance function 
$d(x)=d_{G}(s,x)$ for $x\in V$, and layer function 
$\lambda(x)=\lambda_{G}(x)$ for $x\in V$. Then the following hold:
\begin{enumerate}[label=(\alph*)]
    \item For all $u,v\in V$, 
    $d(u)<d(v) \iff \lambda(u)<\lambda(v)$.
    \item For every edge $(u,v)\in E$:  
    \begin{itemize}
        \item if $(u,v)$ is a back-edge, then $\lambda(v)< \lambda(u)$;  
        \item if $(u,v)$ is a forward-edge, then $\lambda(v)=\lambda(u)+1$. 
    \end{itemize}
\end{enumerate}
\end{lemma}

\begin{proof}

According to the definition, we let
\[
\{ d(x):x\in V \} = \{ p_1,p_2,\ldots,p_k\}, \quad
0=p_1<p_2<\cdots <p_k=d_{G}(s,t).
\]
\medskip
\noindent
\emph{Proof of (a).}  
For any $u,v\in V$, we have
\[
d(u)<d(v) 
\iff p_{\lambda(u)}<p_{\lambda(v)}
\iff \lambda(u)<\lambda(v).
\]

\medskip
\noindent
\emph{Proof of (b).}  
Consider an edge $(u,v)\in E$ with the following two cases: 
\begin{itemize}
\item If $(u,v)$ is a back-edge, 
by~\Cref{def:layered-graph}, we know $d_{G}(s,u)\not =d_{G}(s,v)$ (equivalently, $\lambda(u)\not =\lambda(v)$). Moreover, if $\lambda(u)<\lambda(v)$, we pick an arbitrary vertex $x$ such that $\lambda(x)=\lambda(u)+1$, then 
\[
d_{G}(s,u)<d_{G}(s,x)\le d_{G}(s,v)<d_{G}(s,u)+w(u,v),
\]
contradicting the definition of $(s,t)$-layered graphs.

\item If $(u,v)$ is a forward-edge, then
\[
d(u)+w(u,v)=d(v),
\]
which implies $\lambda(v)>\lambda(u)$.  
Assume for contradiction that $\lambda(v)>\lambda(u)+1$. We choose an arbitrary vertex $x$ such that $\lambda(x)=\lambda(u)+1$.
Then
\[
d_{G}(s,u)<d_{G}(s,x)<d_{G}(s,v)=d_{G}(s,u)+w(u,v),
\]
again contradicting the definition.
\end{itemize}
\end{proof}

According to~\Cref{lem:basic-back-edges}, we know that every next-to-shortest path must contain at least one back-edge. Based on this observation, we define the Back-Edge Decomposition as follows (see \Cref{fig:BEDecomp}): 

\begin{definition}[Back-Edge Decomposition]\label{def:decomp}
Given an $(s,t)$-layered graph $G=(V,E,w)$ with $s,t\in V$, we consider a $\stot$ not-shortest path $P$.  The \textbf{Back-Edge Decomposition} $\BEDec{P}$ is defined as a 5-tuple $(A,B; P_1,P_0,P_2)$, where $A,B\in V$, $P_1=P_{s\to A},P_0=P_{A\to B},P_2=P_{B\to t}$ are paths in $G$, and $P_0$ is the subpath of $P$ starting with the first back-edge of $P$ and ending with the last back-edge.

\end{definition}

\begin{figure}[ht]
    \centering
    \tgraphic{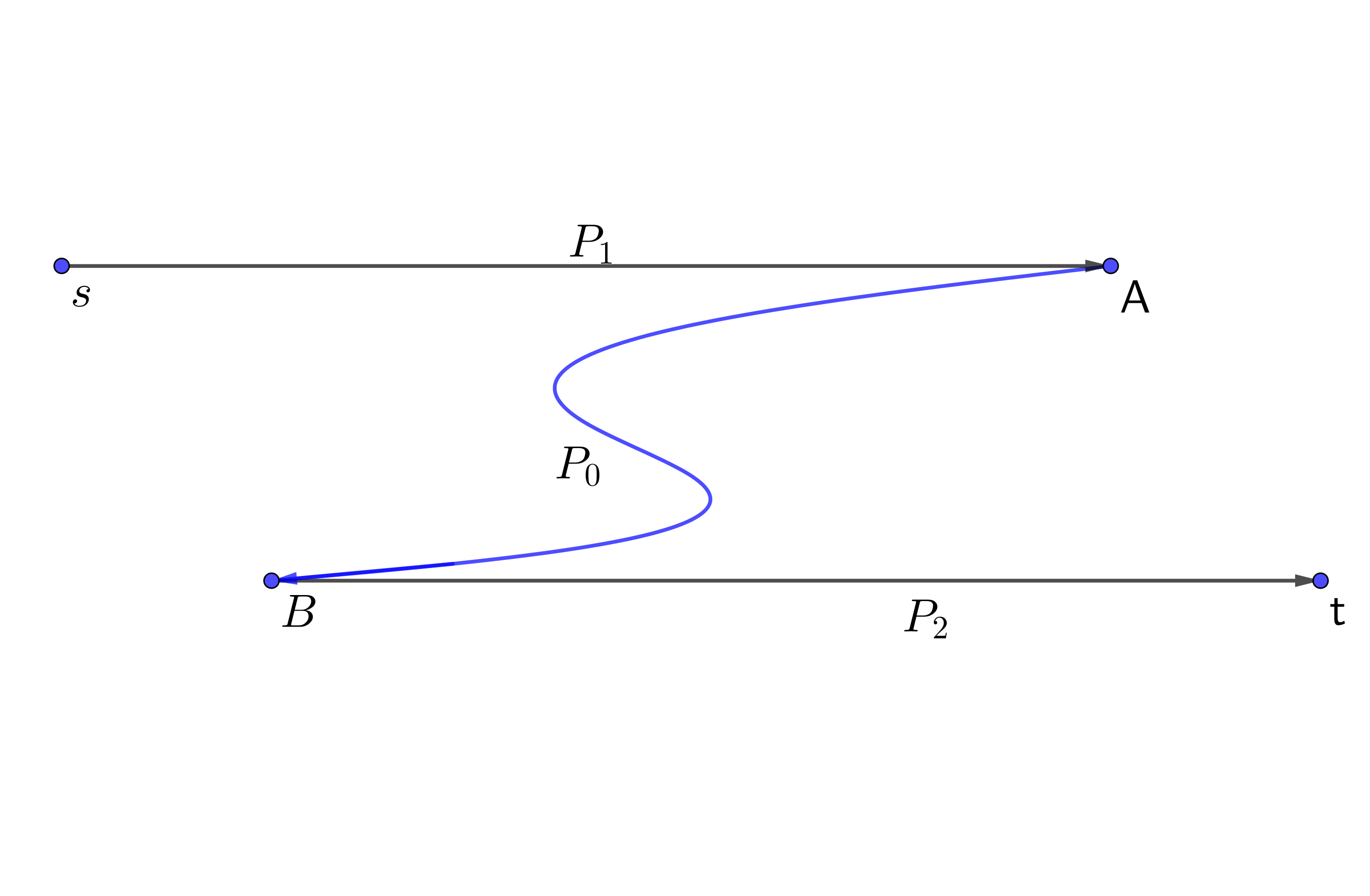}
    \caption{An illustration of Back-Edge Decomposition.}
    \label{fig:BEDecomp}
\end{figure}

 A similar concept ``maximal mixed path", which corresponds to $P_0$ in our~\Cref{def:decomp}, was used in \cite{DBLP:journals/networks/WuW15}. Lemma 1 in~\cite{DBLP:journals/networks/WuW15} implies the following lemma. Here we reprove the lemma using our notation.

\begin{lemma}\label{lem:layer-range}
Given an $(s,t)$-layered graph $G=(V,E,w)$ and any $\stot$ next-to-shortest path $P^*$, let $\BEDec{P^*}=(A,B;P^*_1,P^*_0,P^*_2)$. Then $d(B)<d(A)$, and for all $u\in \VSet(P^*_0)\setminus \lbrace A,B\rbrace$, $d(B)<d(u)<d(A)$. 
\end{lemma}
\begin{proof}
If $\VSet(P^*_0)\setminus \lbrace A,B\rbrace$ is an empty set, then as $P^*_0$ begins with a back-edge, by~\Cref{lem:layered-lambda} (b), $d(B)<d(A)$.
In the following, we assume that $\VSet(P^*_0)\setminus \lbrace A,B\rbrace$ isn't empty.

First we prove that $\forall u\in \VSet(P^*_0)\setminus \lbrace A,B\rbrace$, $d(B)<d(u)$.
We choose $u\in \VSet(P^*_0)\setminus \lbrace A,B\rbrace$ with a minimum layer index. i.e., $u=\arg\min_{u'\in \VSet(P^*_0)\setminus \lbrace A,B\rbrace} \lambda(u')$. 
As $P^*_0$ begins with a back edge, we have $d(u)\leq d(A)$.
Suppose, for contradiction, that $d(u)\le d(B)$. 
Then the forward path $\anyP_{s\to u}$ does not intersect with $(P^*_0)_{u\to B}\circ (P^*_2)_{B\to t}$. 
Let us think about $P'=\anyP_{s\to u}\circ (P^*_0)_{u\to B}\circ (P^*_2)_{B\to t}$. 
We will derive a contradiction by showing that $P'$ is a not-shortest path that is shorter than $P^*$:
\begin{itemize}
\item $w(P')=w(\anyP_{s\to u})+w((P^*_0)_{u\to B})+w((P^*_2)_{B\to t})=d(u)+d(t)-d(B)+w((P^*_0)_{u\to B})$

$\le d(A)+d(t)-d(B)+w((P^*_0)_{u\to B})<w(P^*)$. 
\item Since the last edge of $P^*_0$ must be a back-edge, due to~\Cref{lem:basic-back-edges}, $w(P')>d(t)$. 
\end{itemize}
Therefore, $P'$ is a $\tightto{s}{t}$ not-shortest path that is shorter than $P^*$, which contradicts the assumption that $P^*$ is a next-to-shortest path. Hence, we have $\forall u\in \VSet(P^*_0)\setminus \lbrace A,B\rbrace$, $d(B)<d(u)$.

The rest of the proof is symmetric to the argument above.
We can assume that $v\in V(P^*_0)\setminus \lbrace A,B\rbrace$ is the vertex with maximum $d(v)$.
Suppose for contradiction that $d(v)\ge d(A)$. In this case, by a symmetric argument, $P^*_{s\to v}\circ \anyP_{v\to t}$ is a not-shortest path which is shorter than $P^*$. Consequently, if $\VSet(P_0^{*})\setminus \{A,B\}\neq\emptyset$, we also have $d(B)<d(A)$, since there exists $u\in \VSet(P_0^{*})\setminus \{A,B\}$ with $d(B)<d(u)<d(A)$. 
\end{proof}

The Key Lemma is as follows (see the accompanying \Cref{fig:XY}): 
\begin{lemma}[Key Lemma]\label{lem:key}
Given an $(s,t)$-layered graph $G=(V,E,w)$ and any $\stot$ next-to-shortest path $P^*$, let $\BEDec{P^*}=(A,B;P_1^*,P_0^*,P_2^*)$. There exist two edges $(X',X)\in E,(Y',Y)\in E$ 
satisfying the following conditions:

\begin{itemize}
\item $\lambda(X')=\lambda(Y')=\lambda(X)-1=\lambda(Y)-1$
\item There exists a pair of forward paths $(P_1,P_2)$ such that 
\begin{itemize}
    \item $P_1$ is a $s\utov A$ shortest path that goes through $(X',X)$
    \item $P_2$ is a $B\utov t$ shortest path that goes through $(Y',Y)$ 
    \item $\VSet(P_1)\cap \VSet(P_2)=\emptyset$;
\end{itemize}
and each pair of forward paths $(P_1,P_2)$ that satisfies the above conditions also satisfies:
$$\left(\VSet(P_0^*)\setminus \lbrace A,B\rbrace\right)\cap (\VSet(P_1)\cup \VSet(P_2))=\emptyset.$$
\end{itemize}
\end{lemma}

\begin{figure}[ht]
    \centering
     \tgraphic{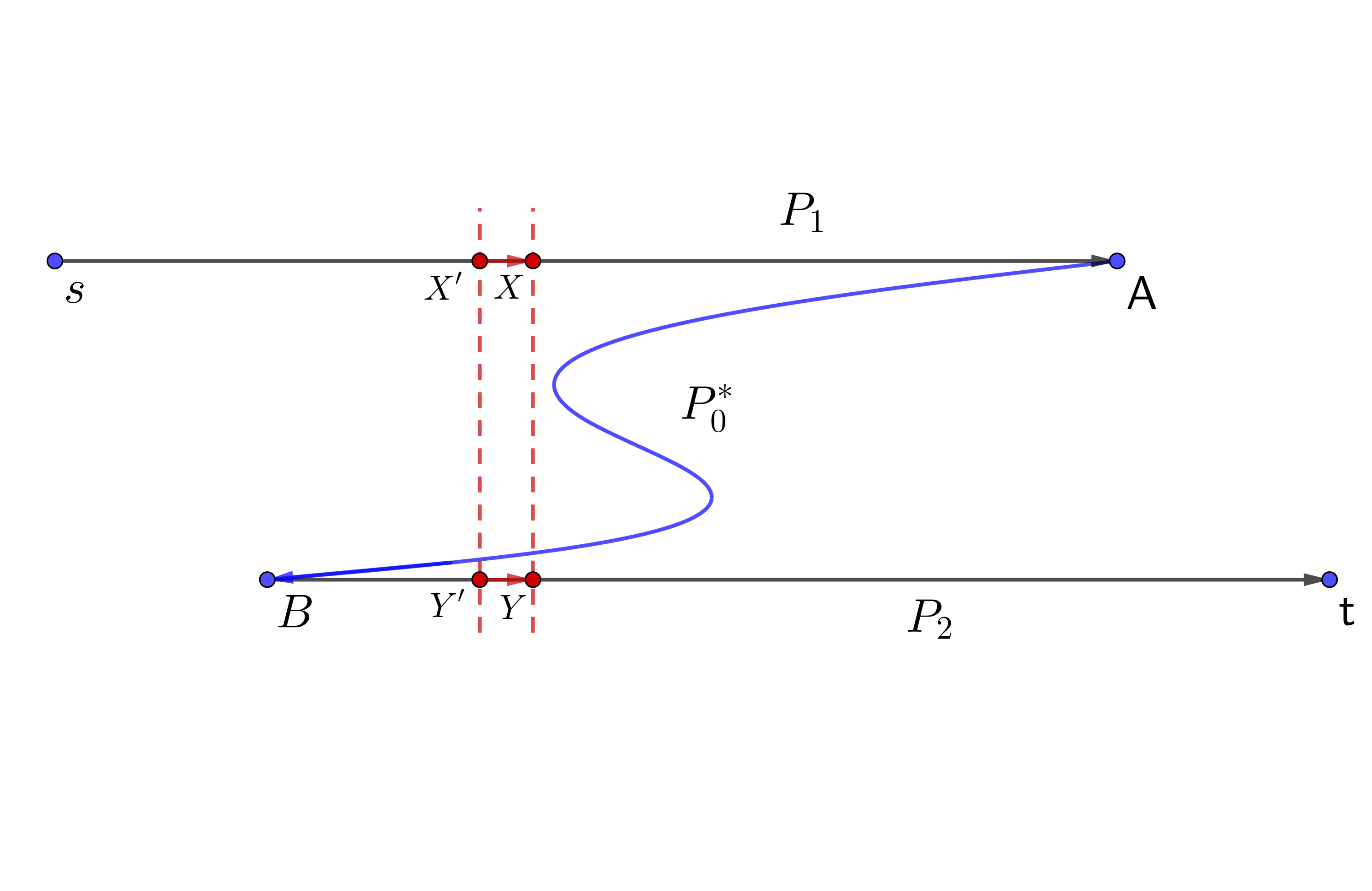}
    \caption{An illustration of $(X',X),(Y',Y)$ and $P^{*}$ in~\Cref{lem:key}}
    \label{fig:XY}
\end{figure}

The proof of the lemma is deferred to \Cref{sec:proof-of-key-lemma}.
Here, based on the Key Lemma, we provide a polynomial-time algorithm $\nspl$. The idea of the algorithm is to try every possible $A,B,X,X',Y,Y'$, and for each such tuple find paths $P_1$, $P_2$ using~\Cref{lem:DAG-VDP}, then combine these paths with a shortest $A\utov B$ path, which is guaranteed to result in a simple path by the Key Lemma.

\pagebreak
\begin{mdframed}
\vspace{-1em}
\paragraph{Algorithm $\nspl(G;s,t)$}\mbox{}\\
\textit{Input:} $(s,t)$-layered graph $G=(V,E,w)$, and two vertices $s,t\in V$ \\
\textit{Output:} A $s\utov t$ next-to-shortest path if one exists, or report $\bot$ otherwise to indicate that no next-to-shortest path exists. 
\begin{enumerate}
    \item Initialize $\hat{P}\gets \bot$.
    \item Enumerate all 6-tuples $(A,B,X',X,Y',Y)\in V^6$ satisfying $d(A)>d(B)$ (equivalently $\lambda(A)>\lambda(B)$), $A,B\in V_{B}(G)$,$(X',X)\in E,(Y',Y)\in E$ and $\lambda(X')=\lambda(Y')=\lambda(X)-1=\lambda(Y)-1$. For each such tuple of vertices, proceed as follows: 
        \begin{enumerate}[label=(\roman*),ref=\theenumi(\roman*)]
            \item Find two disjoint forward paths $P_1,P_2$ satisfying the conditions of~\Cref{lem:key}. 
            Since forward paths contain no back-edges, this reduces to finding two vertex-disjoint paths in a directed acyclic graph (DAG). 
            We can first apply~\Cref{lem:DAG-VDP} to find $(P_1)_{s\to X'},(P_2)_{B\to Y'}$, and then we call the algorithm in~\Cref{lem:DAG-VDP} again to find $(P_1)_{X\to A},(P_2)_{Y\to t}$. After that, we concatenate the output paths with edges $(X',X),(Y',Y)$. In particular, $P_1=(P_1)_{s\to X'}\circ (X',X)\circ (P_1)_{X\to A}$ and $P_2=(P_2)_{B\to Y'}\circ (Y',Y)\circ (P_2)_{Y\to t}$.
            If no such pair $(P_1,P_2)$ exists, proceed to the next iteration with a new tuple $(A,B,X',X,Y',Y)$.  \label{step:P1-P2}
            \item In graph $G\Big\lbrack \Big(V\setminus (\VSet(P_1)\cup \VSet(P_2))\Big)\cup \lbrace A,B\rbrace\Big\rbrack$, we find a shortest $A\utov B$ path $P_0$. (If there is no such path, $P_0\gets \bot$).   \label{step:dij}
            \item If $P_0\not =\bot$ and $w(P_1\circ P_0\circ P_2)<w(\hat{P})$, then we set $\hat{P}\gets P_1\circ P_0\circ P_2$.
        \end{enumerate} \label{step:enumerate}
    \item Return $\hat{P}$.
\end{enumerate}
\end{mdframed}

\begin{theorem}\label{thm:algo-nsp-layered}
Given an $(s,t)$-layered graph $G_L=(V(G_L),E(G_L),w)$ with two distinguished vertices $s,t\in V$, 
there exists a $O(|V_{B}(G_L)|^2|E_{FP}(G_L)|(|E(G_L)|+|V(G_L)|\log |V(G_L)|))$-time algorithm that either finds a $s\utov t$ next-to-shortest path 
(if one exists) or correctly reports that no such path exists. 
\end{theorem}

\begin{proofwithassumption}{\Cref{thm:algo-nsp-layered}}{\Cref{lem:key}}
We consider the above algorithm $\nspl(G_L;s,t)$.
The number of choices for $A,B,(X',X),(Y',Y)$ is $O(|V_{B}(G_L)|^2\cdot |E_{FP}(G_L)|)$. Consider one choice.  
For each choice, by~\Cref{lem:DAG-VDP}, there is a $O(|V(G_L)|+|E(G_L)|)$-time algorithm that can solve 2-VDP in a DAG, so finding $(P_1,P_2)$ takes $O(|V(G_L)|+|E(G_L)|)$ time. To find the shortest from $A$ to $B$, we can use Dijkstra's algorithm, which runs in time $O(|E(G_L)|+|V(G_L)|\log |E(G_L)|)$. Hence, the total time complexity is $O(|V_{B}(G_L)|^2\cdot |E_{FP}(G_L)|\cdot (|E(G_L)|+|V(G_L)|\log |V(G_L)|))$


To prove correctness, we first consider whether a next-to-shortest path exists. 
If none exists, then clearly our algorithm will not find any path. 
Otherwise, let $P^{*}$ be an arbitrary $\stot$ next-to-shortest path with 
$\BEDec{P^{*}}=(A,B;P^{*}_1,P^{*}_0,P^{*}_2)$.

Our algorithm examines all choices of $A,B\in V(G_L)$ and $(X',X),(Y',Y)\in E(G_L)$, 
so there must exist some feasible $A,B,(X',X),(Y',Y)$ satisfying the condition 
in~\Cref{lem:key}. In the corresponding iteration of the loop, 
step~\ref{step:P1-P2} finds two disjoint shortest paths $P_1$ and $P_2$ 
that pass through $(X',X)$ and $(Y',Y)$, respectively. 
By~\Cref{lem:key}, both $P_1$ and $P_2$ do not intersect $P_0^*$ except at the endpoints. 
Therefore, in step~\ref{step:dij}, there must exist an $A\utov B$ path in the remaining graph. 
Let $P_0$ be the path found in step~\ref{step:dij}. 
By the above argument, $P_1\circ P_0\circ P_2$ must be a not-shortest path. In addition,
\[
\begin{aligned}
    w(P_1\circ P_0\circ P_2)
    &= (d(A)-d(s)) + w(P_0) + (d(t)-d(B))\\
    &= w(P_1^{*}) + w(P_0) + w(P_2^{*})\\
    &\le w(P_1^{*}) + w(P_0^{*}) + w(P_2^{*})\\
    &= w(P^{*}),
\end{aligned}
\]
so $P_1\circ P_0\circ P_2$ must be a next-to-shortest path. 
Therefore, if such a next-to-shortest path exists, 
$\nspl(G_L;s,t)$ must return a next-to-shortest path.
\end{proofwithassumption}

Now we restate the main theorem: 
\mainthm*

\begin{proofwithassumption}{\Cref{thm:main}}{\Cref{lem:key}}
    This statement follows by combining~\Cref{thm:reduction} and~\Cref{thm:algo-nsp-layered}.
\end{proofwithassumption}

%% file: key-lemma.tex
\section{Proof of\texorpdfstring{~\Cref{lem:key}}{ Lemma 5.3}} \label{sec:proof-of-key-lemma}
In this section, we provide a proof for the Key Lemma (\Cref{lem:key}). We are given an $(s,t)$-layered graph $G=(V,E,w)$ with vertices $s,t\in V$. To simplify notation, we define the following:
\begin{itemize}
\item For $a,b,c,d\in V$, we call a pair of two paths $(P_1,P_2)$ an \textbf{$(\tightto{a}{b},\tightto{c}{d})$-Pair of Disjoint Forward Paths} ($(\tightto{a}{b},\tightto{c}{d})$-PDFP) if $P_1$ is a $a\utov b$ forward path (i.e., no back-edges), $P_2$ is a $c\utov d$ forward path  and these two paths are vertex-disjoint. 
Recall that such a pair can be found or determined not to exist in linear time by solving the 2-VDP problem on the DAG without back-edges, as stated in~\Cref{lem:DAG-VDP}. 
In addition, if we force these two paths to go through some specific vertices $(u_1,u_2,\ldots,u_p)$ and $(v_1,v_2,\ldots,v_{p})$ respectively, where $p$ is a constant and $\lambda(u_i)=\lambda(v_i)$ for all $1\le i\le p$, we abbreviate the pair as $(a\utov u_1\utov \ldots \utov u_p\utov b,c\utov v_1\utov \ldots \utov v_{p}\utov d)$-PDFP. In this case, we divide the DAG into $p+1$ parts according to the layers and solve $2$-VDP for each part. 
If every part admits such a pair of disjoint paths, then since the parts correspond to disjoint intervals of layers, 
the concatenated paths remain disjoint. Hence, such a pair can also be found or determined not to exist in linear time.  
\item We call an $\tightto{A}{B}$ path $P^{*}_0$ an \textbf{oracle middle path} if there exists an $\stot$ next-to-shortest path $P^{*}$ such that $\BEDec{P^{*}}=(A,B;P_1^*,P_0^*,P_2^*)$ for some other subpaths $P_1^*$,$P_2^*$. In addition, we say that such a pair $(P_1^{*},P_2^{*})$ is \textbf{compatible} with $P_0^*$ (or $P_0^{*}$-compatible).  
\item For an oracle middle path $P_0^{*}$ from $A$ to $B$, and a $(\tightto{a}{b},\tightto{c}{d})$-PDFP $(P_1,P_2)$ ($a,b,c,d\in V$), we denote $\calI(P_0^{*};P_1,P_2)=(\VSet(P_0^{*})\setminus \lbrace A,B\rbrace)\cap (\VSet(P_1)\cup \VSet(P_2))$ as the set of intersection points between $\VSet(P_0^{*})\setminus \lbrace A,B\rbrace$ and the union of $\VSet(P_1)$ and $\VSet(P_2)$.
\item Let $P_0^*$ be an oracle middle path from $A$ to $B$. We say that a 4-tuple $(X',X,Y',Y)\in V^4$ is \textbf{$P_0^*$-isolated} if the following hold: 
\begin{itemize}
    \item $(X',X),(Y',Y)\in E$ and $\lambda(X')=\lambda(Y')=\lambda(X)-1=\lambda(Y)-1$;
    \item There exists a $P_0^{*}$-compatible $(s\utov X'\utov X\utov A,B\utov Y'\utov Y\utov t)$-PDFP;
    \item For all $(s\utov X'\utov X\utov A,B\utov Y'\utov Y\utov t)$-PDFPs $(P_1,P_2)$, $\calI(P_0^*;P_1,P_2)=\emptyset$.
\end{itemize}
In other words, $(X',X,Y',Y)$ satisfies the conditions in~\Cref{lem:key}.
\end{itemize}

Now, we restate~\Cref{lem:key} using the above definitions:
\begin{lemma*}[Restatement of~\Cref{lem:key}]
Given an \textbf{oracle middle path} $P^{*}_0$ from $A$ to $B$, there exists a \textbf{$P_0^*$-isolated} 4-tuple $(X',X,Y',Y)$.
\end{lemma*}

\noindent Here are some high-level ideas for our proof:
\begin{itemize}
    \item We prove the Key Lemma by contradiction. During our proof, we always assume that there exists an $(s\utov X'\utov X\utov A,B\utov Y'\utov Y\utov t)$-PDFP $(P_1,P_2)$ that intersects with $P_0^{*}$ (i.e., $\calI(P_0^{*};P_1,P_2)\not =\emptyset$). 
    \item First, we consider an $s\utov A$ or $B\utov t$ path $P$ that may intersect with $P_0^{*}$. We discuss the relationship between $P$ and any $(P_1^{*},P_2^{*})$ that is compatible with $P_0^{*}$. We claim that the intersection $(\VSet(P_0^{*})\setminus \lbrace A,B\rbrace)\cap \VSet(P)$ consists only of a special type of vertex that we call an \textbf{Up-Vertex}.
    \item Secondly, we explicitly construct a pair $(X,Y)$ such that   $d(B)\le d(X)=d(Y)\le d(A)$, and then further construct $X'$ and $Y'$. We show that if an $(s\utov X'\utov X\utov A,B\utov Y'\utov Y\utov t)$-PDFP $(P_1,P_2)$ intersects with $P_0^{*}$, then we can construct an $(s\utov X\utov A,B\utov Y\utov t)$-PDFP $(P'_1,P_2')$ that intersects $P_0^{*}$ in at least two vertices $\alpha,\beta\in \VSet(P_0^{*})$, where $d(\alpha)<d(X)$ and $d(\beta)>d(X)$.
    \item Finally, based on those two vertices $\alpha$ and $\beta$, we construct a next-to-shortest path with smaller length, which contradicts our assumption that the oracle middle path $P_0^{*}$ corresponds to a next-to-shortest path.
\end{itemize}

\subsection{Relationship between Arbitrary PDFPs and Compatible PDFPs}\label{subsec:pdfps}

The following definition decomposes a path $P$ into segments based on the intersection pattern of $P$ with paths $(P^*_1,P^*_2)$. See \Cref{fig:updown}.

\begin{definition}[Up/Down/Gap Paths and Vertices]\label{def:updown}
Consider an $(s\utov A,B\utov t)$-PDFP $(P^*_1,P^*_2)$ that is compatible with some oracle middle path $P_0^*$. 
For a forward path $P$ with both endpoints in $\VSet(P_1^*)\cup \VSet(P_2^*)$, denote $\interV(P_1^{*},P_2^*;P)=(\VSet(P^{*}_1)\cup \VSet(P^{*}_2))\cap \VSet(P)$. 
According to the order in $P$, we list the items of $\interV(P^*_1,P^*_2;P)$: $(u_1,\ldots,u_{k})$, where $k=|\interV(P^*_1,P^*_2;P)|$ is the number of intersection points.
Then $P$ can be decomposed into $k-1$  subpaths: the $i$-th subpath is $P_{u_i\to u_{i+1}}$ for $1\le i<k$. We classify each subpath $P_{u_i \to u_{i+1}}$ of $P$ into three types, based on the endpoints' membership in $P_1^*$ and $P_2^*$:
\begin{itemize}
  \item \textbf{Up-Path}: $u_i \in \VSet(P_2^*)$, $u_{i+1} \in \VSet(P_1^*)$;
  \item \textbf{Down-Path}: $u_i \in \VSet(P_1^*)$, $u_{i+1} \in \VSet(P_2^*)$;
  \item \textbf{Gap-Path}: both $u_i, u_{i+1} \in \VSet(P_1^*)$ or both in $\VSet(P_2^*)$.
\end{itemize}

For each case, we define the corresponding internal vertices (excluding shared intersection points)
$V(P_{u_i \to u_{i+1}}) \setminus \interV(P_1^*, P_2^*; P)$ 
as \textbf{Up-Vertices}, \textbf{Down-Vertices}, or \textbf{Gap-Vertices}, respectively.

We denote their union over all such subpaths in $P$ by
\[
\upV(P_1^*, P_2^*; P), \quad \downV(P_1^*, P_2^*; P), \quad \text{and } \gapV(P_1^*, P_2^*; P),
\]
respectively. For a PDFP $(P_1, P_2)$, we extend the notation by defining:
\[
\downV(P_1^*, P_2^*; P_1, P_2) := \downV(P_1^*, P_2^*; P_1) \cup \downV(P_1^*, P_2^*; P_2),
\]
and analogously for $\upV$ and $\gapV$.
\end{definition}

\paragraph*{Remark.} The notations $\calI(P_0^{*}; P_1, P_2)$ and $\interV(P_1^{*}, P_2^{*}; P)$ are similar but differ slightly. The former is defined in the context of an oracle middle path $P_0^*$, while the latter arises in the comparison between two PDFPs. Notably, $\interV(P_1^{*}, P_2^{*}; P)$ includes the endpoints of $P$. As a result, we have the following decomposition:
\[
V(P) = \upV(P_1^{*}, P_2^{*}; P) \uplus \downV(P_1^{*}, P_2^{*}; P) \uplus \gapV(P_1^{*}, P_2^{*}; P) \uplus \interV(P_1^{*}, P_2^{*}; P).
\]
\begin{figure}[h!]
    \centering
    \tgraphic{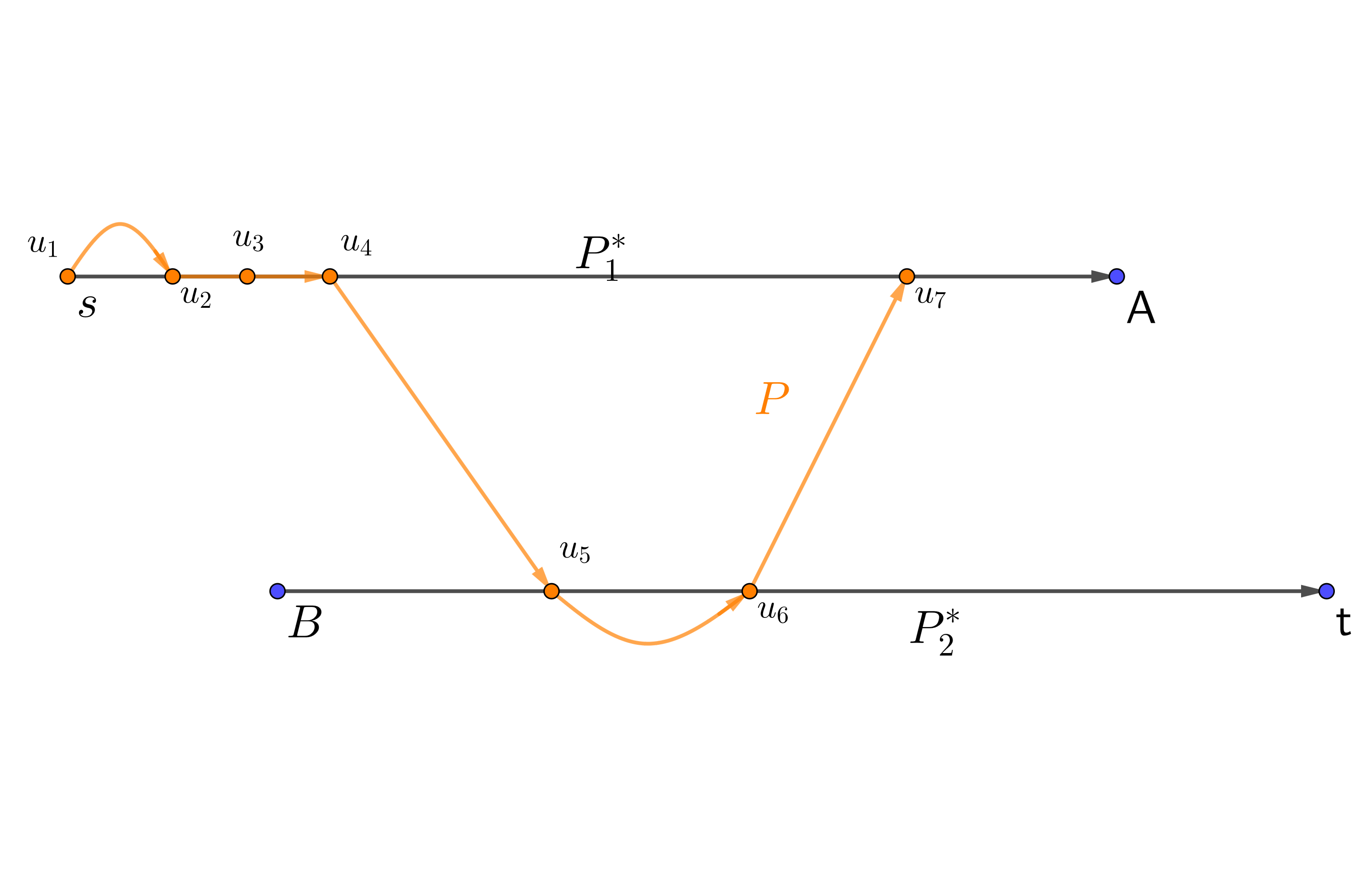}
    \caption{An example of~\Cref{def:updown}. $P$ is a path from $u_1$ to $u_7$, and $\interV(P_1^{*},P_2^{*};P)=\lbrace u_1,u_2,\ldots,u_7\rbrace$. $P_{u_1\to u_2},P_{u_2\to u_3},P_{u_3\to u_4}$ and $P_{u_5\to u_6}$ are gap-paths. However, $P_{u_2\to u_3}$ and $P_{u_3\to u_4}$ contain no gap-vertices. $P_{u_4\to u_5}$ is a down-path. $P_{u_6\to u_7}$ is an up-path.}
    \label{fig:updown}
\end{figure}

\begin{lemma}\label{lem:up-vertex-only}
Given an oracle middle path $P_0^*$ from $A$ to $B$, for any $P_0^{*}$-compatible $(s\utov A,B\utov t)$-PDFP $(P_1^*,P_2^*)$ and any $(s\utov A,B\utov t)$-PDFP $(P_1,P_2)$, $(P_1,P_2)$ intersects $P_0^{*}$ only in up-vertices. i.e.,  
$$\calI(P_0^*;P_1,P_2)\subseteq \upV(P_1^{*},P_2^{*};P_1,P_2).$$
\end{lemma}
\begin{proof}
Since $\calI(P_0^{*};P_1^*,P_2^*)=\emptyset$, for any $u\in \calI(P_0^{*};P_1,P_2)$, $u$ can be an up-vertex, a down-vertex, or a gap-vertex. We argue by contradiction. Suppose, for the sake of contradiction, that $u$ is either a down-vertex or a gap-vertex. 
By~\Cref{def:updown}, this implies that $u$ lies on either a down-path or a gap-path. 
Denote $p$ as this down-path or gap-path, and $u_L,u_R\in \VSet(p)$ are two endpoints of $p$. We discuss two cases based on whether $u_R\in \VSet(P_1^*)$. See \Cref{fig:up-only}. 

\textbf{Case 1: $u_R\in \VSet(P_1^*)$.  }\\ Without loss of generality, we assume that $u$ is the first vertex on $p$ (in the direction from $u_L$ to $u_R$) that is in $P_0^*$. That is, the internal vertices of the subpath $p_{u_L \to u}$ are disjoint from $P_0^*$; formally,
\[
\left( \VSet(p_{u_L \to u}) \setminus \{u_L, u\} \right) \cap \VSet(P_0^*) = \emptyset.
\]
In this case, $u$ must be a gap-vertex and $u_L\in \VSet(P_1^*)$. Therefore, we construct the following walk:
$$W_1=(P_1^*)_{s\to u_L}\circ (p)_{u_L\to u}\circ (P_0^*)_{u\to B}\circ (P_2^*).$$
\begin{itemize}
\item By~\Cref{def:updown}, we know that the subpath $(p)_{u_L \to u}$ does not intersect $P_1^*$ or $P_2^*$, except possibly at its endpoints. Moreover, since $P_0^*$, $P_1^*$, and $P_2^*$ together form a path, they do not intersect one another except at their endpoints. Therefore, this walk contains no repeated vertices.
\item By~\Cref{lem:layer-range}, we know $d(u)>d(B)$, so $(P_0^*)_{u\to B}$ contains a back-edge. By~\Cref{lem:basic-back-edges}, $w(W_1)>d_{G}(s,t)$.
\item 
Recall that $d(u)<d(A)$ by \Cref{lem:layer-range}.
Since $w((p)_{u_L\to u})=d(u)-d(u_L)$ is strictly smaller than $w((P_1^*)_{u_L\to A})$, we have $w(W_1)=d(u)+w\left((P_0^*)_{u\to B}\circ (P_2^*)\right)< w(P_1^{*}\circ P_0^{*}\circ P_2^{*})$. 
\end{itemize}
Therefore, $W_1$ is a shorter $s\utov t$ not-shortest path, which contradicts the assumption that $P_1^{*}\circ P_0^{*}\circ P_2^{*}$ is a next-to-shortest path.

\begin{figure}[t]
  \centering

  \begin{subfigure}{\linewidth}
    \centering
    \tgraphic{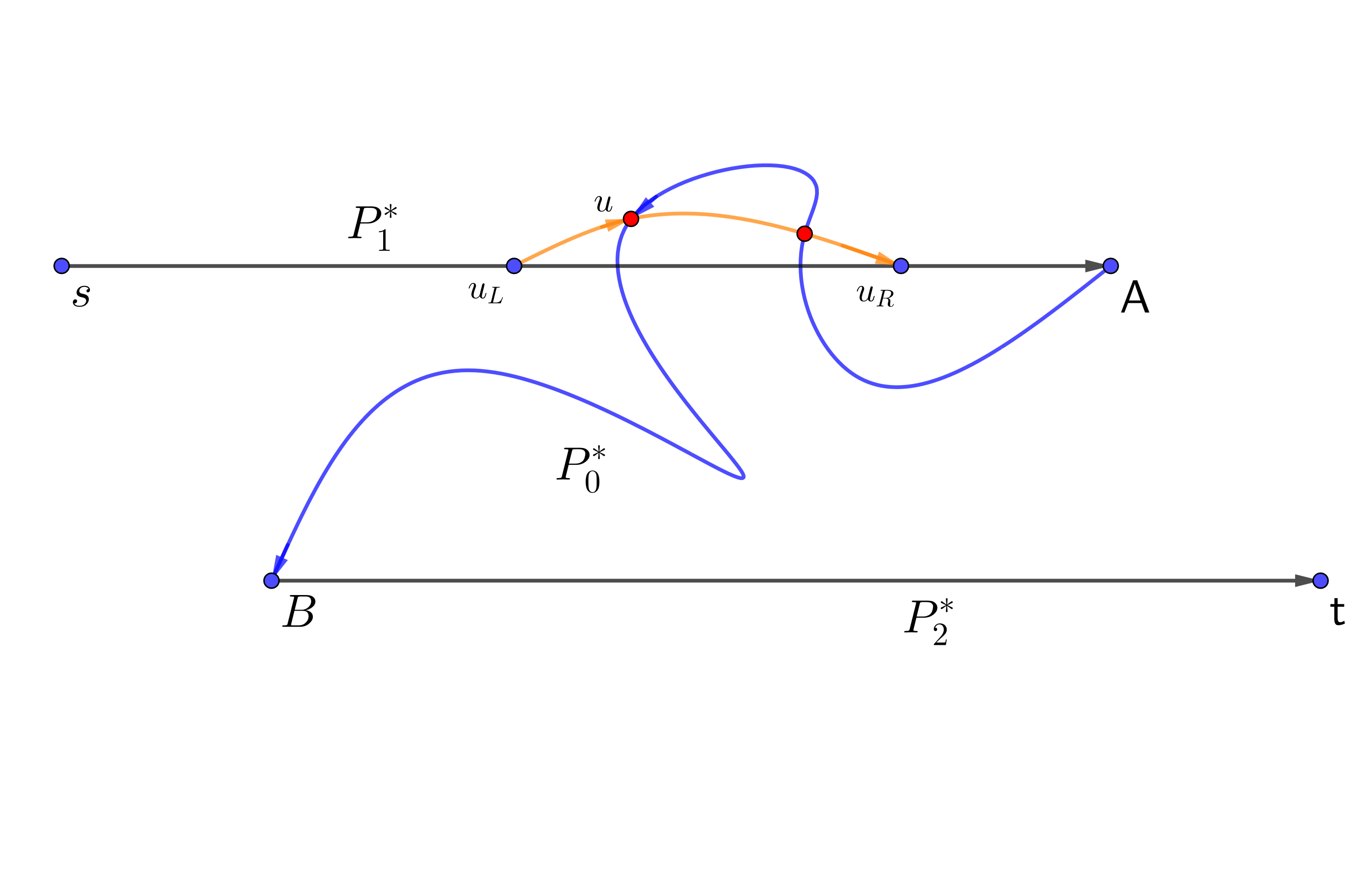}
    \caption{Case 1.}
    \label{fig:up-only-1}
  \end{subfigure}

  \vspace{1em} 

  \begin{subfigure}{\linewidth}
    \centering
    \tgraphic{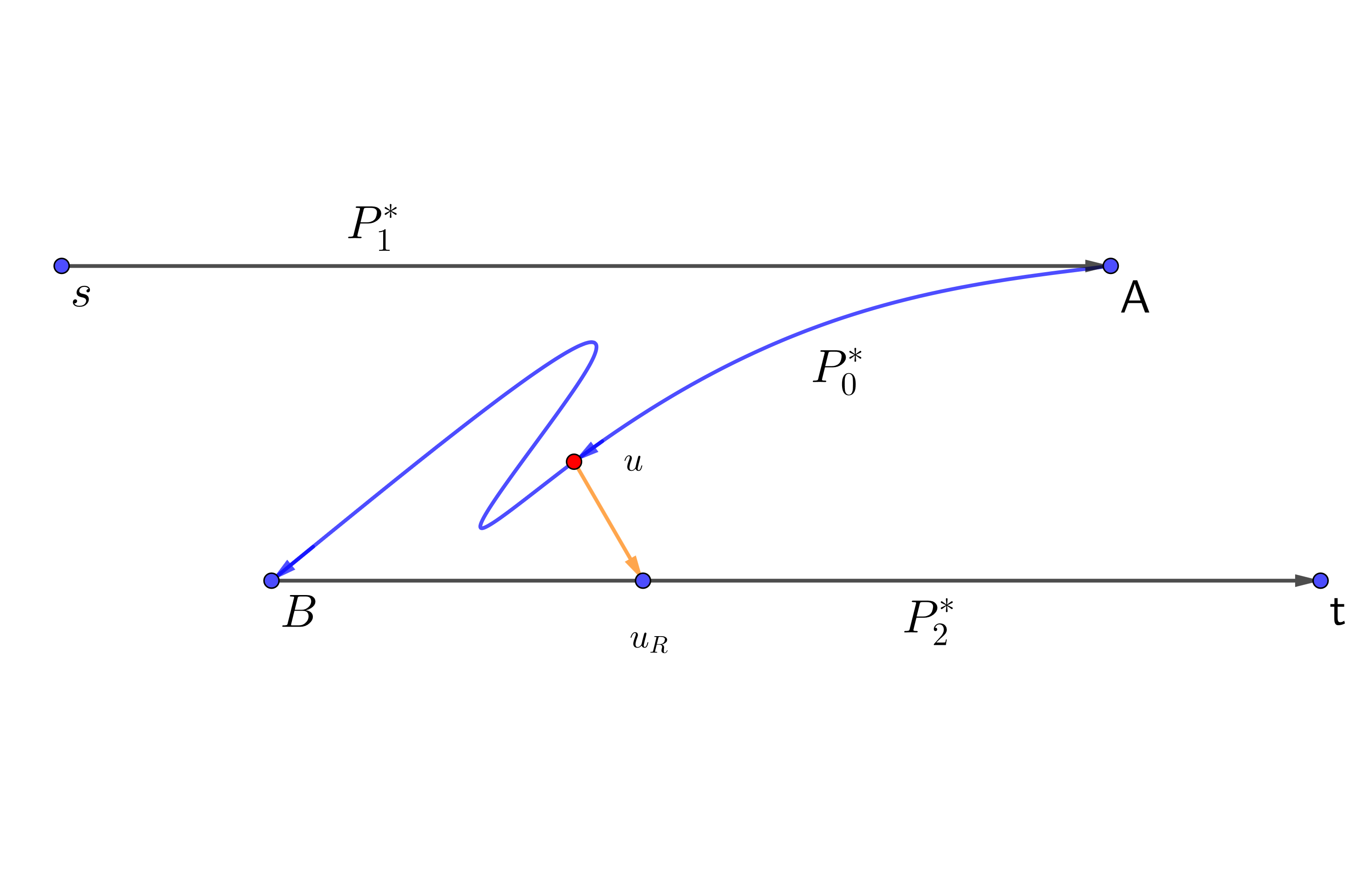}
    \caption{Case 2.}
    \label{fig:sub2}
  \end{subfigure}

  \caption{An illustration of proof for \Cref{lem:up-vertex-only}.}
  \label{fig:up-only}
\end{figure}

\textbf{Case 2: $u_R\in \VSet(P_2^*)$.  }\\ Without loss of generality, we let $u$ be the last vertex on $p$ (in the direction from $u_L$ to $u_R$) that intersects $P_0^{*}$. We consider the following walk:
$$W_2=P_1^{*}\circ (P_0^{*})_{A\to u}\circ p_{u\to u_R}\circ (P_2^{*})_{u_R\to t}.$$
\begin{itemize}
\item By our assumption that $u$ is the last vertex on $p$, $p_{u\to u_R}$ does not intersect $P_0^{*}$ (except for $u$).
Also, by~\Cref{def:updown}, $p_{u\to u_R}$ does not intersect $P_0^{*},P_1^{*}$ or $P_2^{*}$ except at its endpoints, so $W_2$ contains no repeated vertices.
\item By~\Cref{lem:layer-range}, we know $d(u)<d(A)$, so $(P_0^{*})_{A\to u}$ contains a back-edge and $w(W_2)>d_{G}(s,t)$.
\item Similar to case 1, by \Cref{lem:layer-range} we have $d(u)>d(B)$. Therefore,
$w(W_2)=w(P_1^*)+\left(w((P_0^*)_{A\to u})+w(p_{u\to u_R})\right)+w((P_2^*)_{u_R\to t})<w(P_1^*)+w(P_0^*)+w(P_2^*)$.
\end{itemize}
Therefore, $W_2$ is a shorter $s\utov t$ not-shortest path, which contradicts the assumption that $P_1^{*}\circ P_0^{*}\circ P_2^{*}$ is a next-to-shortest path.
\end{proof}

\subsection{Construction of the Isolated Tuple} \label{subsec:construction}
Recall that we are discussing an $(s,t)$-layered graph $G=(V,E,w)$. In this subsection, for an oracle middle path $P_0^{*}$, we provide a construction of the two edges $(X',X),(Y',Y)\in E$, and show some properties of this construction. 

\begin{definition}\label{def:critical}
Given an oracle middle path $P_0^{*}$ from $A$ to $B$, we call a pair $(X,Y)\in V\times V$ \textbf{critical} 
if it satisfies the following conditions:
\begin{enumerate}[label=(\roman*)]
\item $d(B)\le d(X)=d(Y)\le d(A)$
\item There exists a $P_0^{*}$-compatible $(s\utov X\utov A,B\utov Y\utov t)$-PDFP. 
\item There exists a $(s\utov X,B\utov Y)$-PDFP $(P_1',P_2')$ such that $\calI(P_0^{*};P_1',P_2')\not =\emptyset$.
\end{enumerate}
\end{definition}

From~\Cref{def:critical}, it is clear that for any critical pair $(A,B)$, $d(B)>d(X)$; otherwise any $(s\utov X,B\utov Y)$-PDFP would not intersect $P_0^{*}$ by~\Cref{lem:layer-range}.

\begin{lemma}\label{lem:if-no-critical}
Let $G=(V,E,w)$ be an $(s,t)$-layered graph, and $P_0^{*}$ be an oracle middle path from $A$ to $B$. Assume that there is no critical pair. Then for any $P_0^{*}$-compatible $(s\utov A,B\utov t)$-PDFP $(P_1^*,P_2^*)$, and for any tuple $(X',X,Y',Y)$ such that $(X',X)\in \ESet(P_1^*),(Y',Y)\in \ESet(P_2^*)$ and $\lambda(X')=\lambda(Y')$, the 4-tuple $(X',X,Y',Y)$ is $P_0^*$-isolated .
\end{lemma}
\begin{proof}
By definition, we have: 
\begin{itemize}
\item $\lambda(X')=\lambda(Y')=\lambda(X)-1=\lambda(Y)-1$ since $P_1^*$ and $P_2^*$ are forward paths.
\item $(P_1^*,P_2^*)$ is a $P_0^{*}$-compatible $(s\utov X'\utov X\utov A,B\utov Y'\utov Y\utov t)$-PDFP.
\item Consider any $(s\utov X'\utov X\utov A,B\utov Y'\utov Y\utov t)$-PDFP $(P_1,P_2)$. Let $A'$ be the vertex in $P^*_2$ with $d(A')=d(A)$. Assume that $\calI(P_0^*;P_1,P_2)\neq\emptyset$. 
By~\Cref{lem:layer-range}, for all $u\in \calI(P_0^{*};P_1,P_2)$ we have 
$d(B)<d(u)<d(A)$. Hence $(A,A')$ is a critical pair, which contradicts our assumption. 
Therefore, $\calI(P_0^*;P_1,P_2)=\emptyset$.
\end{itemize}
Hence, $(X',X,Y',Y)$ is $P_0^{*}$-isolated.
\end{proof}

Now we provide our construction of $(X',X,Y',Y)$:
\begin{itemize}
    \item If there are no critical pairs, then by the definition of oracle middle path, there exists a $P_0^{*}$-compatible $(s\utov A,B\utov t)$-PDFP $(P_1^*,P_2^*)$. Thus by \Cref{lem:if-no-critical}, the construction of $(X',X,Y',Y)$ is trivial as we can choose any $(X',X)$ and $(Y',Y)$ in $P_1^{*}$ and $P_2^{*}$ respectively with the same layer index. 
    \item Otherwise, there exist at least one critical pairs. In this case, we choose this 4-tuple as follows:
    \begin{itemize}
        \item $(X,Y)$ is an arbitrary leftmost critical pair. i.e., $d(X)$ is minimum among all such pairs.
        \item By~\Cref{def:critical}, we know that there exists a $P_0^{*}$-compatible $(s\utov X\utov A,B\utov Y\utov t)$-PDFP $(P_1^{*},P_2^{*})$. Let $X'$ be the predecessor of $X$ on $P_1^{*}$ and $Y'$ be the predecessor of $Y$ on $P_2^*$, then it is clear that $(X',X),(Y',Y)\in E$ and $\lambda(X')=\lambda(Y')=\lambda(X)-1=\lambda(Y)-1$.
    \end{itemize}
\end{itemize}

\begin{lemma}\label{lem:two-inter}
We are given an $(s,t)$-layered graph $G=(V,E,w)$ and an oracle middle path $P_0^{*}$ from $A$ to $B$. Let $(X',X,Y',Y)$ be as constructed above. Assume that there exists an $(s\utov X'\utov X\utov A, B\utov Y'\utov Y\utov t)$-PDFP $(P_1, P_2)$ such that $\calI(P_0^*; P_1, P_2) \neq \emptyset$. Then there exists an $(s\utov X\utov A, B\utov Y\utov t)$-PDFP $(P_1', P_2')$ that has at least two intersection vertices $\alpha, \beta \in \calI(P_0^*; P_1', P_2')$ with $d(\alpha) < d(X)=d(Y)<d(\beta)$. 
\end{lemma}
\begin{proof}
If there is no critical pair, the statement is trivial: by~\Cref{lem:if-no-critical}, there exists no $(s\utov X'\utov X\utov A, B\utov Y'\utov Y\utov t)$-PDFP that intersects $P_0^*$. 

In the remainder of the proof, we may therefore assume that a critical pair exists, and $(X,Y)$ is a leftmost one. Let $u\in \calI(P_0^{*};P_1,P_2)$ be an arbitrary intersection vertex. We discuss the following three cases: 
\begin{itemize}
\item \textbf{Case 1:} $\lambda(u)<\lambda(X)$.\\
In this case, $(X',Y')$ is a critical pair and $\lambda(X')<\lambda(X)$. This contradicts the assumption that $(X,Y)$ is leftmost.
\item \textbf{Case 2:} $\lambda(u)=\lambda(X)$.\\
It is impossible because $(P_1,P_2)$ only goes through $(X,Y)$ in layer $\lambda(X)$.
\item \textbf{Case 3:} $\lambda(u)>\lambda(X)$.\\
Since $(X,Y)$ is a critical pair (\Cref{def:critical}), there exists a $(s\utov X,B\utov Y)$-PDFP $(P_1'',P_2'')$ that intersects $P_0^{*}$. 
In other words, there exists some $\alpha\in \calI(P_0^*;P_1'',P_2'')$ such that $d(\alpha)<d(X)$ (equivalently, $\lambda(\alpha)<\lambda(X)$). 

Recall $\lambda(X)=\lambda(Y)$. We consider an $(s\utov X\utov A, B\utov Y\utov t)$-PDFP
$$(P_1',P_2')=(P_1''\circ (P_1)_{X\to A},P_2''\circ (P_2)_{Y\to t}).$$ It contains at least two intersection vertices $\alpha$ and $\beta=u$, where $\lambda(\alpha)<\lambda(X)$ and $\lambda(\beta)>\lambda(Y)$. This completes the proof. 
\end{itemize}
\end{proof}
\subsection{Final Analysis}\label{subsec:final}

 In this subsection, we will complete the proof of the Key Lemma (\Cref{lem:key}). We need to show that the choice of $(X',X,Y',Y)$ in~\Cref{subsec:construction} is $P_0^{*}$-isolated. According to~\Cref{lem:two-inter}, we only need to consider the case that there exists an $(s\utov X\utov A,B\utov Y\utov t)$-PDFP with at least two intersection vertices with $P_0^{*}$. 
 To prove this, we focus on the structural relationship between an arbitrary PDFP and a $P_0^*$-compatible PDFP. 

\begin{lemma}\label{lem:no-gap}
Given an $(s,t)$-layered graph $G=(V,E,w)$, an oracle middle path $P_0^{*}$, and an arbitrary corresponding critical pair $(X,Y)$, for any $(s\utov X\utov A,B\utov Y\utov t)$-PDFP $(P_1,P_2)$, there exists a $P_0^{*}$-compatible $(s\utov X\utov A,B\utov Y\utov t)$-PDFP $(P_1^{**},P_2^{**})$  such that: $$\gapV(P_1^{**},P_2^{**};P_1,P_2)=\emptyset$$ 
\end{lemma}
\begin{proof}
Assume that we have an $(s\utov X\utov A,B\utov Y\utov t)$-PDFP $(P_1,P_2)$.
We choose $(P_1^{**},P_2^{**})$ as an arbitrary $P_0^{*}$-compatible $(s\utov X\utov A,B\utov Y\utov t)$-PDFP that maximizes:
$$|\left(\ESet(P_1^{**})\cup \ESet(P_2^{**})\right)\cap \left(\ESet(P_1)\cup \ESet(P_2)\right)|$$

i.e., the number of common edges. 

We prove the statement by contradiction. Suppose that $\gapV(P_1^{**},P_2^{**};P_1,P_2)\not =\emptyset$, i.e.,~there exists a vertex $u\in \gapV(P_1^{**},P_2^{**};P_1,P_2)$. We know $u$ is in $P_1$ or $P_2$. Denote this path as $P_i$ for a specific $i\in \lbrace 1,2\rbrace$, and let $u_L,u_R\in V$ be the two endpoints of this gap-path. 

Let $u_L,u_R\in \VSet(P_j^{**})$ for a specific $j\in \lbrace 1,2\rbrace$. We can replace $(P_j^{**})_{u_L\to u_R}$ by $(P_i)_{u_L\to u_R}$. 

\begin{itemize}
\item By~\Cref{lem:up-vertex-only}, $(P_i)_{u_L\to u_R}$ does not intersect $P_0^{*}$. i.e., $\VSet((P_i)_{u_L\to u_R})\cap \left(\VSet(P_0^*)\setminus \lbrace A,B\rbrace\right)=\emptyset$. In addition, according to the definition of gap-path, $(P_i)_{u_L\to u_R}$ does not intersect $P^{**}_{3-j}$, and neither $X$ nor $Y$ is an internal vertex of $(P_i)_{u_L\to u_R}$. Therefore, after replacing $(P_j^{**})_{u_L\to u_R}$ by $(P_i)_{u_L\to u_R}$, the paths $P_1^{**}$, $P_2^{**}$
are $P_0^{*}$-compatible $(s\utov X\utov A,B\utov Y\utov t)$-PDFP.
\item The replacement provides more common edges: 
$$\begin{aligned}
    |\ESet((P_i)_{u_L\to u_R})\cap \left(\ESet(P_1)\cup \ESet(P_2)\right)|&=\lambda(u_R)-\lambda(u_L)\\
    &>|\ESet((P^{**}_j)_{u_L\to u_R})\cap \left(\ESet(P_1)\cup \ESet(P_2)\right)|
\end{aligned}$$
\end{itemize}

This contradicts the assumption that $(P_1^{**}$, $P_2^{**})$ is chosen to maximize the number of common edges with $(P_1,P_2)$. 
\end{proof}

Based on a critical pair $(X,Y)$, we define two regions with respect to $\lambda(X)$:
\begin{itemize}
  \item the left region: $\{ u \in V : \lambda(u) < \lambda(X) \}$,
  \item the right region: $\{ u \in V : \lambda(u) > \lambda(X) \}$.
\end{itemize}

\begin{lemma}
\label{lem:up-path-exclusive}
Let $G=(V,E,w)$ be an $(s,t)$-layered graph, and let $P_0^{*}$ be an oracle middle path with critical pair $(X,Y)$.  
For any $(s\utov X\utov A,B\utov Y\utov t)$-PDFP $(P_1,P_2)$ and $P_0^{*}$-compatible $(s\utov X\utov A,B\utov Y\utov t)$-PDFP $(P_1^{**},P_2^{**})$  such that $\gapV(P_1^{**},P_2^{**};P_1,P_2)=\emptyset$, the following hold:
\begin{enumerate}[label=(\alph*)]
\item For any up-path $p$, all vertices of $p$ except the endpoints lie in exactly one of the two regions; that is, either $d(u)\le d(X)$ for all $u\in \VSet(p)$, or $d(u)\ge d(X)$ for all $u\in \VSet(p)$.
\item If $P_1$ or $P_2$ contains an up-path in the right region, then $P_1$ contains a down-path $u_{1L}\to u_{1R}$ in the right region and $P_2$ contains an up-path $u_{2L}\to u_{2R}$ in the right region. Moreover, $\lambda(u_{1R})>\lambda(u_{2L})$ and $\lambda(u_{2R})>\lambda(u_{1L})$.
\end{enumerate}
\end{lemma}

\begin{proof}
\mbox{}\\
\noindent
\emph{Proof of (a).} Since $P_1$ and $P_1^{**}$ are forward paths and both pass through $X$, any up-path in $P_1$ must be contained entirely in either the left region or the right region, except for the endpoints. Similarly, since $P_2$ and $P_2^{**}$ are forward paths and both pass through $Y$, any up-path in $P_2$ must be contained entirely in either the left region or the right region, except for the endpoints.

\medskip
\noindent
\emph{Proof of (b).}
As $P_1$ or $P_2$ contains an up-path in the right region, we have $(P_1)_{X\to A}\neq (P_1^{**})_{X\to A}$ and $(P_2)_{Y\to t}\neq (P_2^{**})_{Y\to t}$ (because otherwise $P_1$ and $P_2$ would intersect).
Notice that there are no gap paths in $P_1$ and $P_2$.
Therefore, in the right region, $P_1$ contains a down-path, and $P_2$ contains an up-path.

Let the first down path in $P_1$ in the right region be from $u_{1L}$ to $u_{1R}$, and the first up-path of $P_2$ in the right region be from $u_{2L}$ to $u_{2R}$. Here ``first path" means the one with the leftmost starting vertex.
Since $P_1$ and $P_2$ are forward paths, this is an intuitive definition. 
As $u_{1R}\not\in P_2,u_{1R}\in P_2^{**}$ and $(P_2)_{Y\to u_{2L}}\subseteq P_2^{**}$, we must have $\lambda(u_{1R})>\lambda(u_{2L})$ to avoid intersection between $P_1$ and $P_2$.
Similarly, we have $\lambda(u_{2R})>\lambda(u_{1L})$.
\end{proof}


\begin{proofoflem}{\ref{lem:key}}
We prove the Key Lemma by contradiction.

Recall that we are given an $(s,t)$-layered $G = (V, E,w)$ and an oracle middle path $P_0^*$. We construct the 4-tuple $(X', X, Y', Y)$ as described in~\Cref{subsec:construction}. That is, by~\Cref{lem:if-no-critical}, we may assume that critical pairs exist, and we fix a leftmost critical pair $(X, Y)$. Then we select a $P_0^*$-compatible $(s \utov X \utov A, B \utov Y \utov t)$-PDFP $(P_1^*, P_2^*)$, and choose edges $(X', X) \in \ESet(P_1^*)$ and $(Y', Y) \in \ESet(P_2^*)$. 

Suppose for contradiction that there exists an $(s\utov X'\utov X\utov A,B\utov Y'\utov Y\utov t)$-PDFP $(P_1,P_2)$ that intersects $P_0^{*}$. Then by~\Cref{lem:two-inter}, there exists an $(s\utov X\utov A,B\utov Y\utov t)$-PDFP $(P_1',P_2')$ with two intersection vertices in the left region and the right region, respectively. Applying~\Cref{lem:no-gap} to $(P_1',P_2')$, we obtain a $P_0^{*}$-compatible $(s\utov X\utov A,B\utov Y\utov t)$-PDFP $(P_1^{**},P_2^{**})$ such that $\gapV(P_1^{**},P_2^{**};P_1',P_2')=\emptyset$.

Let $u_{\text{first}}$ be the first vertex in $\calI(P_0^{*};P_1',P_2')$ with respect to $P_0^{*}$ (in the direction from $A$ to $B$). 
Based on the layer index of $u_{\text{first}}$, we distinguish the following two cases (we know $d(u_{\text{first}})\not =d(X)$ since $u_{\text{first}}$ is either in $P_1'$ or $P_2'$). An illustration of the paths and vertices introduced in each of the two cases is shown in \Cref{fig:final-1} and \Cref{fig:final-2}. 

\begin{figure}[ht]
    \centering
    \tgraphic{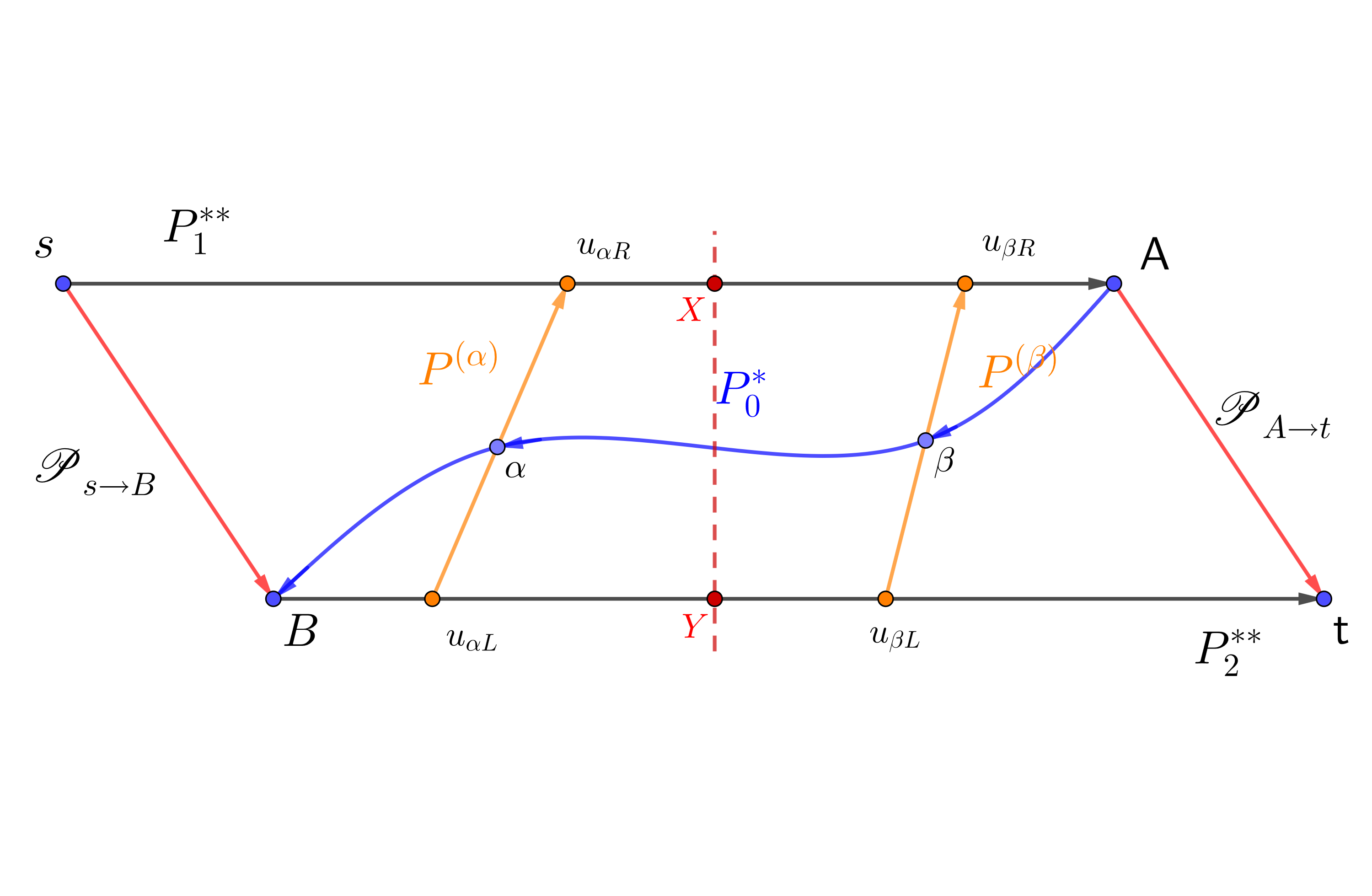}
    \caption{An illustration of Case 1 in the proof of~\Cref{lem:key}. }
    \label{fig:final-1}
\end{figure}

\textbf{Case 1: }$d(u_{\text{first}})>d(X)$.\\
In this case, by~\Cref{lem:two-inter}, there must exist another intersection vertex in the left region. Hence, there exist two consecutive intersection vertices $\beta,\alpha$ (in the order along $P_0^{*}$) such that $\beta$ lies in the right region and $\alpha$ lies in the left region. Then we have: $(\VSet((P_0^{*})_{\beta\to \alpha})\setminus \lbrace \alpha,\beta\rbrace)\cap (\VSet(P'_1)\cup \VSet(P'_2))=\emptyset$.  By~\Cref{lem:up-vertex-only}, we know both $\alpha$ and $\beta$ are in up-paths, so we introduce the following notation:
\begin{itemize}
\item Let $\alpha$ be in the up-path $P^{(\alpha)}$, and the two endpoints of this up-path are $u_{\alpha L}$ and $u_{\alpha R}$.
\item Let $\beta$ be in the up-path $P^{(\beta)}$, and the two endpoints of this up-path are $u_{\beta L}$ and $u_{\beta R}$.
\end{itemize}
Recall that $\anyP_{s\to B}$ and $\anyP_{A\to t}$ are arbitrary but fixed shortest $s\to B$ and $A\to t$ paths, respectively. Now we consider the following walk:

$$\begin{aligned}W_2^*=\anyP_{s\to B}\circ (P^{**}_2)_{B\to u_{\beta L}}\circ (P^{(\beta)})_{u_{\beta L}\to \beta}\circ (P_0^{*})_{\beta\to \alpha}\\\circ (P^{(\alpha)})_{\alpha\to u_{\alpha R}}\circ (P_1^{**})_{u_{\alpha R}\to A}\circ \anyP_{A\to t}\end{aligned}$$

To obtain a contradiction, we will show that $W_2^{*}$ is an $s\utov t$ not-shortest path that is shorter than $P_1^{**}\circ P_0^{*}\circ P_2^{**}$.
\begin{itemize}
    \item (No Repeated Vertices) We can notice the followings:
    \begin{itemize}
        \item By~\Cref{lem:layer-range}, $\anyP_{s\to B}$ and $\anyP_{A\to t}$ are disjoint and they do not intersect other fragments. 
        \item $\beta,\alpha$ are consecutive intersection vertices, so $(P_0^{*})_{\beta\to \alpha}$ does not intersect the fragments that are subpaths of $P^{(\alpha)},P^{(\beta)}$ except the endpoints $\lbrace \alpha,\beta\rbrace$. In addition, $(P_0^{*})_{\beta\to \alpha}$ does not intersect the fragments that are subpaths of $P^{**}_1,P^{**}_2$ because those are compatible with $P^*_0$. Hence, $(P_0^{*})_{\beta\to \alpha}$ does not intersect the rest of $W_2^{*}$. 
        \item $P^{(\alpha)}$ and $P^{(\beta)}$ do not intersect each other because they are different parts of the same PDSP. Also, their internal vertices do not intersect $P_1^{**}$ or $P_2^{**}$ since they are up-paths.
    \end{itemize}
    \item (Not Shortest Path) $d(\beta)>d(X)>d(\alpha)$, so $W_2^*$ includes a back-edge.
    \item (Shorter than $P_1^{**}\circ P_0^{*}\circ P_2^{**}$) We can notice that $d(\alpha)>d(B)$ and $d(\beta)<d(A)$ (\Cref{lem:layer-range}). Therefore, since $W^*_2$ consists of a forward path from $s$ to $\beta$, a subpath of $P^*_0$, and a forward path from $\alpha$ to $t$, 
    $$\begin{aligned}w(W_2^{*}) & \le d(\beta)+w(P_0^{*})+d(t)-d(\alpha)\\
    &<d(A)+w(P_0^{*})+d(t)-d(B)=w(P_1^{**}\circ P_0^{*}\circ P_2^{**})\end{aligned}$$
\end{itemize}

\begin{figure}[ht]
    \centering
    \tgraphic{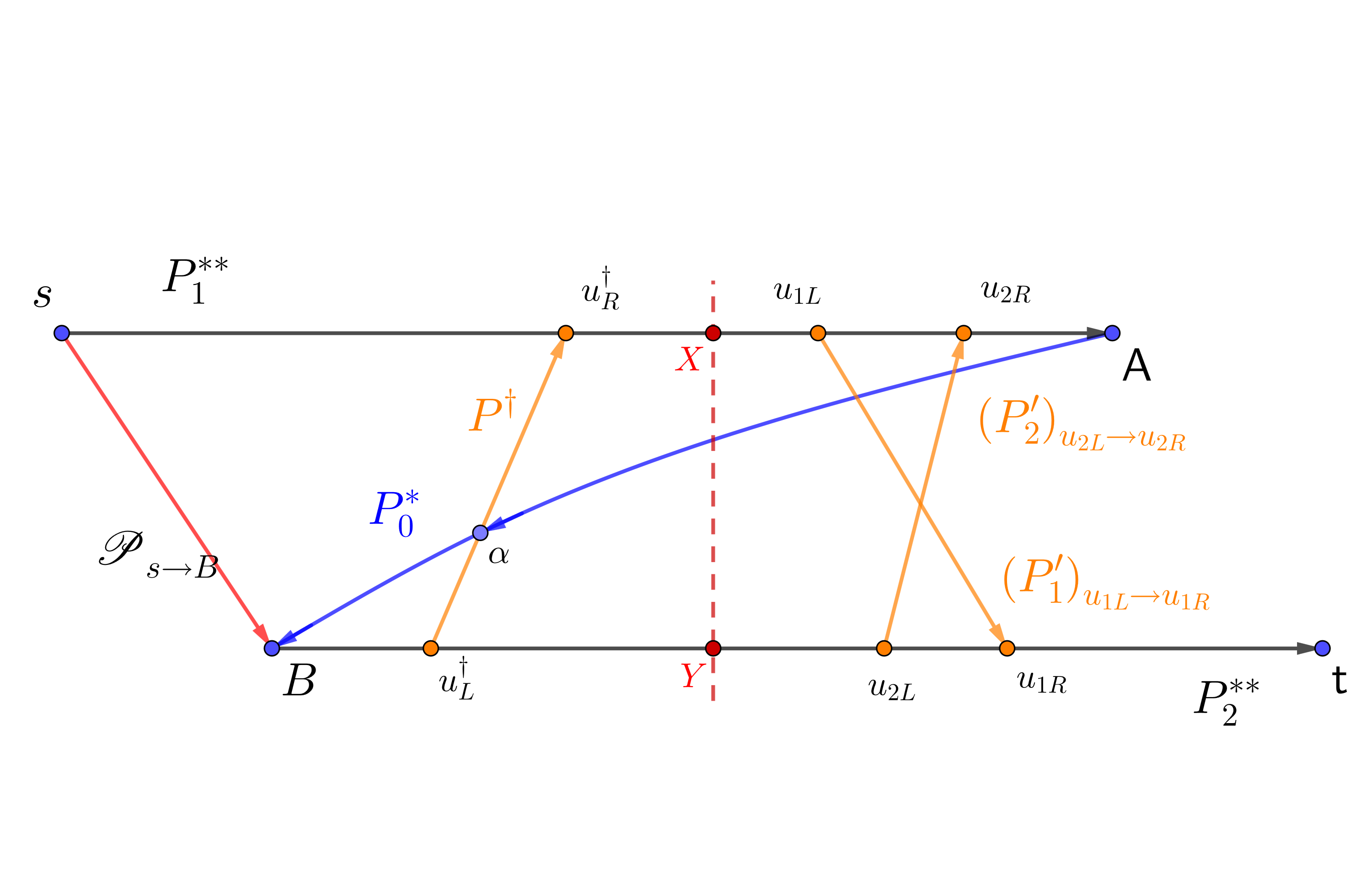}
    \caption{An illustration of Case 2 in the proof of~\Cref{lem:key}. }
    \label{fig:final-2}
\end{figure}

\textbf{Case 2: } $d(u_{\text{first}})<d(X)$.\\
In this case, we cannot apply the same reasoning as in Case 1. Specifically, if $P_0^*$ traverses the layers in a “squiggly’’ manner such that all intersection points in the left region appear before those in the right region along $P_0^*$, then the fragment $(P_0^*)_{\beta \to \alpha}$ no longer exists.

Here we let $\alpha=u_{\text{first}}$, and let $\beta$ be an arbitrary intersection vertex in the right region (i.e., $d(\beta)>d(X)$). By~\Cref{lem:up-vertex-only}, $\beta$ must be an up-vertex (i.e., $\beta\in \upV(P_1^{**},P_2^{**};P'_1,P_2')$). Therefore, by \Cref{lem:up-path-exclusive}, there exists an up-path that entirely lies in the right region. 
  
By \Cref{lem:up-path-exclusive}, there exist a down-path of $P_1'$ in the right region from $u_{1L}$ to $u_{1R}$, and an up-path of $P_2'$ in the right region  from $u_{2L}$ to $u_{2R}$. 
Similarly, according to~\Cref{lem:up-vertex-only}, we know that $\alpha$ is in an up-path. We denote this up-path as $P^{\dagger}$, and let it be from $u^{\dagger}_{L}\in \VSet(P_2^{**})$ to $u^{\dagger}_{R}\in \VSet(P_1^{**})$. By~\Cref{lem:up-path-exclusive}(a), we know that $P^{\dagger}$ lies entirely in the left region except for the endpoints, and $\lambda(u^{\dagger}_{R})\le \lambda(X)$. Recall that $\anyP_{s\to B}$ is an arbitrary but fixed forward path from $s$ to $B$. Then we consider the following walk:

$$\begin{aligned}W^{*}_1=\anyP_{s\to B}\circ (P_2^{**})_{B\to u_{2L}}\circ (P_2')_{u_{2L}\to u_{2R}}\circ (P_1^{**})_{u_{2R}\to A}\circ (P_0^{*})_{A\to \alpha}\\ 
\circ (P^{\dagger})_{\alpha\to u^{\dagger}_{R}}\circ (P_1^{**})_{u^{\dagger}_{R}\to u_{1L}}\circ (P_1')_{u_{1L}\to u_{1R}}\circ (P_2^{**})_{u_{1R}\to t}\end{aligned}$$

 and for a concatenation $W=W_1\circ \ldots \circ W_{k}$, we call $W_1,\ldots,W_{k}$ \textbf{fragments} of $W$. To obtain a contradiction, we will show that $W_1^{*}$ is an $s\utov t$ not-shortest path that is shorter than $P_1^{**}\circ P_0^{*}\circ P_2^{**}$. 
\begin{itemize}
\item (No Repeated Vertices)  According to~\Cref{lem:layer-range}, $\anyP_{s\to B}$ doesn't intersect the other fragments of $W_1^{*}$ because all vertices in other fragments are of layer index not smaller than $B$.  In addition, since $\alpha=u_{\text{first}}$ is the first intersection vertex, $(P_0^{*})_{A\to \alpha}$ does not intersect the remaining fragments. The fragments that are subpaths of $P^{**}_1$ and $P^{**}_2$ are disjoint from each other by the fact that $u_{1L}$ comes before $u_{2R}$ on $P^{**}_1$, and that $u_{2L}$ comes before $u_{1R}$ on $P^{**}_2$. These fragments are also disjoint from the internal vertices of the fragments that are subpaths of $P'_1,P'_2$, since the latter are up-paths and down-paths. So, it remains to show that subpaths of $P'_1,P'_2$ are disjoint from each other. $(P^{\dagger})_{\alpha\to u^{\dagger}_{R}}$, $(P'_1)_{u_{1L}\to u_{1R}}$, and $(P'_2)_{u_{2L}\to u_{2R}}$ are mutually disjoint, since $P^{\dagger}$ is an up-path and any up-path does not intersect other up-paths or down-paths except at the endpoints. 
\item (Not Shortest Path) By \Cref{lem:layer-range}, $d(u_{\text{first}})<d(A)$, so $W_1^{*}$ includes a back-edge.
\item (Shorter than $P_1^{**}\circ P_0^{*}\circ P_2^{**}$) We can notice that $d(\alpha)>d(u^{\dagger}_{L})\ge d(B)$. Since $W^*_1$ consists of a forward path from $s$ to $A$, a subpath of $P^*_0$, and a forward path from $u_{\text{first}}$ to $t$, we have: $$\begin{aligned}w(W_1^{*}) & \le d(A)+w(P_0^{*})+d(t)-d(u_{\text{first}})\\
&<d(A)+w(P_0^{*})+d(t)-d(B)=w(P_1^{**}\circ P_0^{*}\circ P_2^{**}).\end{aligned}$$
\end{itemize}

To summarize, in both \textbf{Case 1} and \textbf{Case 2}, we derived a contradiction by constructing a not-shortest path that is shorter than $P_1^{**}\circ P_0^{*}\circ P_2^{**}$, contradicting the assumption that $P_0^*$ is an oracle middle path. This completes the proof of~\Cref{lem:key} by showing that $(X',X,Y',Y)$, as constructed in~\Cref{subsec:construction}, is $P_0^{*}$-isolated. 
\end{proofoflem}